\pdfoutput=1
\pdfoutput=1
\pdfoutput=1
\pdfoutput=1
\pdfoutput=1
\documentclass[11pt]{article}
\usepackage{subfloat}
\usepackage{amssymb}
\usepackage{amsmath}
\usepackage{amsthm}
\usepackage{array}
\usepackage{graphicx}
\usepackage{subfigure}
\usepackage{epstopdf}
\usepackage{cases}
\usepackage{color}
\usepackage{multirow}
\usepackage{picinpar}
\usepackage[colorlinks,
            linkcolor=red,
            anchorcolor=green,
            citecolor=blue
            ]{hyperref}
\usepackage{floatrow}
\usepackage{enumitem}
\usepackage{fourier}
\begin{document}

\newcommand{\defi}{\stackrel{\Delta}{=}}
\newcommand{\A}{{\cal A}}
\newcommand{\B}{{\cal B}}
\newcommand{\U}{{\cal U}}
\newcommand{\G}{{\cal G}}
\newcommand{\cZ}{{\cal Z}}
\newcommand\one{\hbox{1\kern-2.4pt l }}
\newcommand{\Item}{\refstepcounter{Ictr}\item[\left(\theIctr\right)]}
\newcommand{\QQ}{\hphantom{MMMMMMM}}

\newtheorem{Theorem}{Theorem}[section]
\newtheorem{Lemma}{Lemma}[section]
\newtheorem{Corollary}{Corollary}[section]
\newtheorem{Remark}{Remark}[section]
\newtheorem{Example}{Example}[section]
\newtheorem{Proposition}{Proposition}[section]
\newtheorem{Property}{Property}[section]
\newtheorem{Assumption}{Assumption}[section]
\newtheorem{Definition}{Definition}[section]
\newtheorem{Construction}{Construction}[section]
\newtheorem{Condition}{Condition}[section]
\newtheorem{Exa}[Theorem]{Example}
\newcounter{claim_nb}[Theorem]
\setcounter{claim_nb}{0}
\newtheorem{claim}[claim_nb]{Claim}
\newenvironment{cproof}
{\begin{proof}
 [Proof.]
 \vspace{-3.2\parsep}}
{\renewcommand{\qed}{\hfill $\Diamond$} \end{proof}}
\newcommand{\erhao}{\fontsize{21pt}{\baselineskip}\selectfont}
\newcommand{\xiaoerhao}{\fontsize{18pt}{\baselineskip}\selectfont}
\newcommand{\sanhao}{\fontsize{15.75pt}{\baselineskip}\selectfont}
\newcommand{\sihao}{\fontsize{14pt}{\baselineskip}\selectfont}
\newcommand{\xiaosihao}{\fontsize{12pt}{\baselineskip}\selectfont}
\newcommand{\wuhao}{\fontsize{10.5pt}{\baselineskip}\selectfont}
\newcommand{\xiaowuhao}{\fontsize{9pt}{\baselineskip}\selectfont}
\newcommand{\liuhao}{\fontsize{7.875pt}{\baselineskip}\selectfont}
\newcommand{\qihao}{\fontsize{5.25pt}{\baselineskip}\selectfont}
\newcounter{Ictr}
\renewcommand{\theequation}{
\arabic{equation}}
\renewcommand{\thefootnote}{\fnsymbol{footnote}}

\def\A{\mathcal{A}}

\def\C{\mathcal{C}}

\def\V{\mathcal{V}}

\def\I{\mathcal{I}}

\def\Y{\mathcal{Y}}

\def\X{\mathcal{X}}

\def\J{\mathcal{J}}

\def\Q{\mathcal{Q}}

\def\W{\mathcal{W}}

\def\S{\mathcal{S}}

\def\T{\mathcal{T}}

\def\L{\mathcal{L}}

\def\M{\mathcal{M}}

\def\N{\mathcal{N}}
\def\R{\mathbb{R}}
\def\H{\mathbb{H}}

\title{}
\author{}
\begin{center}
\topskip2cm
\LARGE{\bf Safe subspace screening for the adaptive nuclear norm regularized trace regression}
\end{center}

\begin{center}
\renewcommand{\thefootnote}{\fnsymbol{footnote}}Pan Shang,  Lingchen Kong \footnote{Address: Pan Shang  and Lingchen Kong are with the School of Mathematics and Statistics, Beijing Jiaotong University, Beijing, 100044, China.\\
E-mail: pshang@amss.ac.cn,   konglchen@126.com }\\
\today
\end{center}
\vskip4pt
\textbf{Abstract:} Matrix form data sets arise in many areas, so there are lots of works about the matrix regression models. One special model of these models  is the adaptive nuclear norm regularized trace regression, which has been proven have good statistical performances. In order to accelerate the computation of this model, we consider the technique called screening rule. According to matrix decomposition and optimal condition of the model, we develop a safe subspace screening rule that can be used to identify inactive subspace of the solution decomposition and reduce the dimension of the solution. To evaluate the efficiency of the safe subspace screening rule, we embed this result into the alternating direction method of multipliers algorithm under a sequence of the tuning parameters. Under this process, each solution under the tuning parameter provides a matrix decomposition space. Then, the safe subspace screening rule is applied to eliminate inactive subspace, reduce the solution dimension and accelerate the computation process. Some numerical experiments are implemented on simulation data sets and real data sets, which illustrate the efficiency of our screening rule.
\section{Introduction}
To deal with matrix form data sets, there recently have some works about the matrix regression models, including multivariate linear regression, trace regression and so on. See e.g., \cite{B08, E18, H22, K11,  M21, L10, L12, N11, Y07, Z14, Z20, Z22}. Unlike the classical linear regression models, the linear coefficient of matrix regression models are in matrix form.  Among these models, there is a special model named the adaptive nuclear norm regularized trace regression \cite{B08}. This model can be treated as a matrix form of the adaptive lasso in Zou \cite{Z06} and is proved consistent to solution rank. Due to this advantage, we consider this model in the paper. Note that this model is a convex problem and can be solved by some existing solvers. The matrix dimension and singular value decomposition is a main concern, therefore it may consume more computational cost  during the process of solving this model.

Recently, there is a computational technique called screening rule, that tends to eliminate inactive features or samples in high-dimensional data sets and accelerate the calculation process. According to our best of knowledge, there are two types screening rules: heuristic screening rule and safe screening rule. We focus on safe screening rules in this paper. The safe screening rule was proposed by Ghaoui et al. in \cite{G12}, where they use the duality theorem of convex problems to identify inactive features and inactive sample in sparse linear regression models and classification problems, respectively. Then, there are lots of safe screening rules proposed, see, e.g.,  \cite{E23,D21,K23,L14,L17,N17,R17,P22,S22,S16,W15,W15b,W19,X16,Z17}, which include results for sparse linear regression models and sparse classification models. Due to the fact we consider a regression model in this paper, we review related works as follows. \cite{L14,W15,X16,L17,R17} build up safe feature screening rules for the famous lasso via different optimization techniques. \cite{W15b} considers the fused lasso only with fused term under the help of variational inequality. \cite{N17,D21} applies the duality gap to set up the screening rules for group lasso. \cite{W19} uses the result of conjugate function to build up the screening rule for sparse group that is the actually first work about the model with two regularizer. \cite{S22} proposes the dual circumscribed sphere technique and set up the screening rule for sparse quantile regression. However, as far as we know, there are a few screening rules for matrix regression models. \cite{WY15} considers the multivariate linear regression with $\ell_{2,1}$ norm as regularizer, and builds up the inactive column screening rule with the help of the optimal condition. \cite{ZZ15} focuses on the multivariate linear regression with nuclear norm as regularizer, and transfers the matrix as sum of rank one matrixes to identify inactive subspace. \cite{SK21} builds up the screening rule for trace regression models with nuclear norm regularizer and a general convex loss function, which identifies inactive singular values of the solution and estimate the upper bound of the solution rank. To characterize the low rank property of the matrix data, comparing to $\ell_{2,1}$ norm, nuclear norm is a direct choice because that it is the convex relaxation of the rank function. In addition, the inactive singular value identification screening rule can not reduce the computation cost.

In this paper, we build up the safe subspace screening rule for adaptive nuclear norm regularized trace regression. This screening rule can deal with the difficulty of the nuclear norm and reduce the computation of this model via eliminating inactive subspaces and reducing dimension of the solution. We apply the matrix decomposition in  \cite{ZZ15} to transfer the solution as sum of rank one matrixes. By analyzing the dual form and Karush-Kuhn-Tucker (KKT) condition of the model, we identify the rank one matrixes with zero coefficient. The zero coefficient means that the corresponding subspace is not involved in the solution decomposition, which is the core purpose of our screening rule. This screening rule has a closed-form formulation and can be implemented with low computational cost. To solve the adaptive nuclear norm regularized trace regression, we present the detailed process of alternating direction method of multipliers (ADMM). By embedding the proposed safe subspace screening rule into ADMM under a sequence of tuning parameters, as showed in Algorithm 2, we can solve the solution path of the adaptive nuclear norm regularized trace regression with low computational cost. We verify the efficiency of the proposed screening rule on some data sets, which clarifies that our result can reduce the calculation time.

Our paper is organized as follows. We present notations in the next paragraph. In Section 2, we introduce the adaptive nuclear norm regularized trace regression model, and review related definitions and results. Section 3 presents the detailed process of building up the safe subspace screening rule with its implementation algorithm. In Section 4, numerical experiments on some data sets are showed. Conclusion remark is presented in Section 5.

\emph{Notations}:
Let $M\in\mathbb{R}^{p\times q}$ be any matrix. The notation vec$(M)$ denotes the vector in $\mathbb{R}^{pq}$ obtained by stacking its columns into a vector. For $j\in\{1,2,\cdots,p\}$ and  $k\in\{1,2,\cdots,q\}$, $M_{:,j}$ means the $j_{th}$ column  and $M_{k,:}$ means the $k_{th}$ row of $M$. Suppose $M$ has a singular value decomposition with nonincreasing singular values $\sigma_{1}(M)\geq \cdots \geq \sigma_{r}(M)>0$ and $r\leq\rm{min\left\{p, q\right\}}$ is the rank of $M$. There are some related  norms with singular values of $M$. The Frobenius norm $\|\cdot\|_{F}$ is defined as $\|M\|_{F}=\sqrt{\sum_{i=1}^{p}\sum _{j=1}^{q}M_{i,j}^{2}}=\sqrt{\sigma^{2}_{1}(M)+\cdots+\sigma^{2}_{r}(M)}$. The nuclear norm $\|\cdot\|_{*}$ is  the sum of non-zero singular values, i.e., $\|M\|_{*}=\sum_{i=1}^{r}{\sigma_{i}(M)}$. The spectral norm $\|\cdot\|_{2}$ is the largest singular value, i.e., $\|M\|_{2}=\sigma_{1}(M)$.  A symmetric matrix $M\in\mathbb{R}^{p\times p}$ is called positive semidefinite (positive definite), denoted as $M\succeq0(M\succ0)$, if $\textbf{\textit{x}}^{\top}M\textbf{\textit{x}}\geq0$ $ (\textbf{\textit{x}}^{\top}M\textbf{\textit{x}}>0)$ holds for any $0\neq\textbf{\textit{x}}\in\mathbb{R}^{p}$. Any $M\in\mathbb{O}^{p}$ means that $M\in\mathbb{R}^{p\times p}$ and $M^{\top}M=I_{p}$. For any vector $\textbf{\textit{x}}\in\mathbb{R}^{p}$,
the norm $\|\cdot\|$ is defined as $\|\textbf{\textit{x}}\|=\sqrt{\sum_{i=1}^{p}x_{i}^{2}}$. The notation $\textbf{\textit{x}}\in R^{n}_{++}$ means all elements of $\textbf{\textit{x}}$ are positive. The indictor function of a set $\mathcal{A}$ is denoted as $\delta_{\mathcal{A}}(\cdot)$, which means the value of this function is zero when variable in the set $\mathcal{A}$ and is $\infty$ otherwise. For any index set $I$, its cardinal number  is denoted as $|I|$ that counts the number in the index set $I$, and its complementary set is denoted as $I^{c}$.  In this paper, the notation 0 may represent scalar, vector and matrix, which can be inferred from the context.
\section{Preliminaries}
This section presents the adaptive nuclear norm regularized trace regression model and review some related basic results.
\subsection{Model Analysis}
\label{2.1}
The statistical model of the trace regression is
\begin{eqnarray*}
y = \langle X,B^{*}\rangle+\epsilon,
\end{eqnarray*}
where $X\in \mathbb{R}^{p\times q}$ is the prediction variable, $y\in \mathbb{R}$ is the response variable, $\epsilon\in \mathbb{R}$ is a random error and $B^{*}\in \mathbb{R}^{p\times q}$ is the true coefficient. By sampling $n$ times, we get
\begin{eqnarray*}
y_{i} = \langle X_{i},B^{*}\rangle+\epsilon_{i}, \quad i=1,2,\cdots,n.
\end{eqnarray*}
Let $\mathcal{X}=(\rm{vec}(\textit{X}_{1}),\cdots,\rm{vec}(\textit{X}_{\textit{n}}))^{\top}\in\mathbb{R}^{\textit{n}\times \textit{p}_{1} \textit{p}_{2}}$, $\boldsymbol{y}=(y_{1},\cdots,y_{n})^{\top}\in\mathbb{R}^{n}$ and $\varepsilon=(\epsilon_{1},\cdots,\epsilon_{n})^{\top}\in\mathbb{R}^{n}$. The sample model can be written as
\begin{eqnarray*}
\boldsymbol{y} = \mathcal{X}\rm{vec}(\textit{B}^{*})+\varepsilon.
\end{eqnarray*}
To estimate the unknown matrix $B^{*}$, there are some literatures about the nuclear norm regularized trace regression, such as Koltchinskii et al. \cite{K11},  Negahban and Wainwright \cite{N11}, Zhou and Li \cite{Z14} and so on, which is
\begin{eqnarray*}
\underset{B\in \mathbb{R}^{p\times q}}\min\left\{\frac{1}{2n}\sum\limits_{i=1}^{n}\left(y_{i}-\langle X_{i},B\rangle\right)^{2}+\lambda\|B\|_{*}\right\}.
\end{eqnarray*}
Here, $\lambda>0$ is a tuning parameter. However, this model  can hardly achieve the rank consistent solution in high-dimensional case as in Bach \cite{B08}. So, the adaptive nuclear norm regularized trace regression  was proposed in Bach \cite{B08} as
\begin{eqnarray}\label{eq1}
\underset{B\in \mathbb{R}^{p\times q}}\min\left\{\frac{1}{2n}\sum\limits_{i=1}^{n}\left(y_{i}-\langle X_{i},B\rangle\right)^{2}+\lambda\|W_{1}BW_{2}\|_{*}\right\},
\end{eqnarray}
where $W_{1}$ and $W_{2}$ are given weight matrixes based on the solution of least square trace regression. The detailed results of these matrixes are showed as follows.
Let
$$\hat{B}_{LS}=\underset{B\in \mathbb{R}^{p\times q}}{\arg\min}\left\{\frac{1}{2n}\sum\limits_{i=1}^{n}\left(y_{i}-\langle X_{i},B\rangle\right)^{2}\right\}.$$
Suppose $\hat{B}_{LS}$ has full rank. Let the singular value decomposition of $\hat{B}_{LS}$ be $\hat{B}_{LS}=U_{LS}$Diag$(s^{LS})V^{\top}_{LS}$,
where $U_{LS}\in\mathbb{O}^{p}$, $V_{LS}\in\mathbb{O}^{q}$, $s^{LS}\in \mathbb{R}^{\min\{p,q\}}$
and Diag$(s^{LS})$ is a matrix whose diagonal vector is the singular vector $s^{LS}$. $s^{LS}$ can be completed by $n^{-1/2}$ to reach dimensions $p$ or $q$. Then,
\begin{center}
$W_{1}=U_{LS}\rm{Diag}(\textit{s}^{\textit{LS}})^{-\gamma}$$\textit{U}^{\top}_{\textit{LS}}$, $W_{2}=V_{LS}\rm{Diag}(\textit{s}^{\textit{LS}})^{-\gamma}$$\textit{V}^{\top}_{\textit{LS}}$ , $\gamma\in(0,1]$.
\end{center}
Based on expressions  of  $W_{1}$ and $W_{2}$, they are obvious symmetric matrixes, i.e., $W^{\top}_{1}=W_{1}$ and
$W^{\top}_{2}=W_{2}$. In addition, $W^{-1}_{1}=U_{LS}\rm{Diag}(\textit{s}^{\textit{LS}})^{\gamma}$$\textit{U}^{\top}_{\textit{LS}}$ and $W^{-1}_{2}=V_{LS}\rm{Diag}(\textit{s}^{\textit{LS}})^{\gamma}$$\textit{V}^{\top}_{\textit{LS}}$.

\subsection{Basic  Results}
\label{2.2}
In this section, some related results are  proposed and reviewed, which provide theoretical guarantees for the rest of this paper. The following definitions are from Rockafellar \cite{R70}.
\begin{Definition}
Let $f:\mathbb{R}^{p\times q}\rightarrow (-\infty,+\infty]$ be a proper closed convex function and let $M\in \mathbb{R}^{p\times q}$.
The  subdifferential of $f$ at $M$ is denoted by $\partial f(M)$ and
\begin{center}
$\partial f(M)=\left\{G\in \mathbb{R}^{p\times q}: f(N)\geq f(M)+\langle G,N-M\rangle, \forall N\in \mathbb{R}^{p\times q} \right\}$.
\end{center}
\end{Definition}
According to this definition, we compute the subdifferential of $f(M)=\lambda\|W_{1}MW_{2}\|_{*}$ as follows.
\begin{Proposition}
For any $M\in\mathbb{R}^{p\times q}$, let $r$ be the rank of $W_{1}MW_{2}$. The subdifferential of $\|W_{1}MW_{2}\|_{*}$ is
\begin{center}
$\partial \|W_{1}MW_{2}\|_{*}=\left\{ W_{1}(U_{r}V^{\top}_{r}+N)W_{2}:W_{1}MW_{2}=U_{r}\rm{Diag}(\sigma_{\textit{r}})\textit{V}^{\top}_{\textit{r}},\|\textit{N}\|_{2}\leq1,
\textit{U}^{\top}_{\textit{r}}\textit{N}=0,\textit{NV}_{\textit{r}}=0\right\}.$
\end{center}
\end{Proposition}
\begin{proof}
It is clear that $f(M)=h(g(M))$ with $g(M)=W_{1}MW_{2}$ and $h(M)=\lambda\|M\|_{*}$. Because $g$ is a linear transformation,
according to \cite[Theorem 3.43]{B17},
\begin{center}
$\partial f(M)=g^{\top}\left(\partial_{g(M)}h(g(M))\right)$
\end{center}
with $g^{\top}$ being the adjoint operator of $g$.

According to the definition of the adjoint operator, $\langle g(M),N\rangle=\langle M,g^{\top}(N)\rangle $, which leads to $g^{\top}(N)=W^{\top}_{1}NW^{\top}_{2}=W_{1}NW_{2}$.

Based on  the singular value decomposition of $W_{1}MW_{2}$ and the subdifferential of the nuclear norm in Watson \cite{W92}, it holds that
\begin{center}
$\partial_{g(M)}h(g(M))=\left\{U_{r}V^{\top}_{r}+N:\|N\|_{2}\leq1,U^{\top}_{r}N=0,NV_{r}=0\right\}$.
\end{center}
Combing these results, the desired result can be obtained.
\end{proof}
The conjugate function is the basis of duality theory, which is the main theoretical guarantee in Section 3, so we review it in the next definition.
\begin{Definition}
Let $f: \mathbb{R}^{p\times q}\rightarrow (-\infty,+\infty]$ be a proper closed convex function. For any $M\in \mathbb{R}^{p\times q}$,
the conjugate function $f^{*}: \mathbb{R}^{p\times q}\rightarrow \mathbb{R}$ is defined as
$$f^{*}(M)= \underset {N\in \mathbb{R}^{p\times q}} \sup\left\{\langle M,N\rangle-f(N)\right\}.$$
\end{Definition}
According to this definition, we compute the conjugate function of $f(M)=\|W_{1}MW_{2}\|_{*}$ as follows.
\begin{Proposition}
For any $M\in\mathbb{R}^{p\times q}$, the conjugate function of  $\|W_{1}MW_{2}\|_{*}$ is
\begin{center}
$(\|W_{1}MW_{2}\|_{*})^{*}=\delta_{\{\|W^{-1}_{1}\cdot W^{-1}_{2}\|_{2}\leq 1\}}(M)$.
\end{center}
\end{Proposition}
\begin{proof}
Based on the definition of the conjugate function,
\begin{equation*}
\begin{aligned}
(\|W_{1}MW_{2}\|_{*})^{*}&=\underset{N\in\mathbb{R}^{p\times q}}\max\left\{\langle N,M\rangle-\|W_{1}NW_{2}\|_{*}\right\}\\
&=\underset{N\in \mathbb{R}^{p\times q},C\in \mathbb{R}^{p\times q}}{\max}\left\{\langle N,M\rangle-\|C\|_{*}\right\}\\
&\quad\quad\quad s.t. \quad C-W_{1}NW_{2}=0.
\end{aligned}
\end{equation*}
The Lagrangian function of this model is
\begin{equation*}
\begin{aligned}
\mathcal{L}(N,C;\Lambda)&=\langle N,M\rangle-\|C\|_{*}+\langle \Lambda,C-W_{1}NW_{2}\rangle\\
&=\left\{\langle N,M\rangle-\langle \Lambda,W_{1}NW_{2}\rangle\right\}+\left\{\langle \Lambda,C\rangle-\|C\|_{*}\right\}.
\end{aligned}
\end{equation*}
Then,
\begin{equation*}
\begin{aligned}
\underset{N,C}\max\mathcal{L}(N,C;\Lambda)&=\underset{N}\max\left\{\langle N,M\rangle-\langle \Lambda,W_{1}NW_{2}\rangle\right\}+\underset{C}\max\left\{\langle \Lambda,C\rangle-\|C\|_{*}\right\}\\
&=\delta_{\{W_{1}\Lambda W_{2}\}}(M)+\delta_{\{\|\cdot\|_{2}\leq1\}}(\Lambda)\\
&=
\begin{cases}
0,& M=W_{1}\Lambda W_{2}, \|\Lambda\|_{2}\leq1\\
\infty, &\rm{otherwise},
\end{cases}
=\delta_{\{\|W^{-1}_{1}\cdot W^{-1}_{2}\|_{2}\leq1\}}(M),
\end{aligned}
\end{equation*}
where the second equality is due to the conjugate function of nuclear norm, i.e., $(\|\Lambda\|_{*})^{*}=\delta_{\{\|\cdot\|_{2}\leq1\}}(\Lambda)$. Note that $0\leq(\lambda\|W_{1}MW_{2}\|_{*})^{*}\leq\underset{N,C}\max\mathcal{L}(N,C;\Lambda)$. The desired result can be obtained. 
\end{proof}
There is another widely used function, called the proximal mapping.
\begin{Definition}
Let $f:\mathbb{R}^{p\times q}\rightarrow (-\infty,+\infty]$ be a proper closed convex function. The proximal mapping of $f$ under any $M\in \mathbb{R}^{p\times q}$
is  given by
\begin{center}
prox$_{f}(M)=\underset{N\in \mathbb{R}^{p\times q}}{\arg\min}\left\{f(N)+\frac{1}{2}\|N-M\|^{2}_{2}\right\}$.
\end{center}
\end{Definition}
From the result in Beck \cite{B17}, we review the proximal mapping of $f(M)=\|M\|_{*}$ as follows.
\begin{Example}
For any $M\in\mathbb{R}^{p\times q}$. The proximal mapping of $\|M\|_{*}$ is
\begin{center}
prox$_{\|\cdot\|_{*}}(M)=U\rm{Diag}((\sigma(\textit{M})-1)_{+})$$\textit{V}^{\top}$,
\end{center}
where $M=U\rm{Diag}(\sigma(\textit{M}))$$\textit{V}^{\top}$ and $(x)_{+}$ means the maximal value of zero and $x$.
\end{Example}

\section{Safe Subspace Screening Rule}
In this section, we provide a safe accelerate technique for (\ref{eq1}). The computational algorithm is the basic method to compute the solution of (\ref{eq1}) under any tuning parameter. To obtain the best effect of this model, we must select the best tuning parameter, which is usually selected from a solution path of (\ref{eq1}). Therefore, we design a safe subspace screening rule to accelerate the computation of the solution path.
\subsection{Matrix Decomposition}
In this section, we review some related facts about the matrix space. For any matrix $B\in\mathbb{R}^{p\times q}$, this matrix can be expressed as a linear combination of $pq$ rank one matrices, i.e.,
\begin{center}
$B=\sum\limits_{j=1}^{p}\sum\limits_{k=1}^{q} B_{j,k}\boldsymbol{e}_{j}\tilde{\boldsymbol{e}}^{\top}_{k}$,
\end{center}
where $\boldsymbol{e}_{j}\in\mathbb{R}^{p}$ with its $j_{th}$ element is one and other elements are zero, $\tilde{\boldsymbol{e}}_{k}\in\mathbb{R}^{q}$ with its $k_{th}$  elements is one and others are zero.  In addition, $\{\boldsymbol{e}_{j}\}_{j=1}^{p}$ is a standard orthogonal basis of the vector space
 $\mathbb{R}^{p}$ and  $\{\tilde{\boldsymbol{e}}_{k}\}_{k=1}^{q}$ is a standard orthogonal basis of the vector space $\mathbb{R}^{q}$. In fact, suppose  $\{\boldsymbol{u}_{j}\}_{j=1}^{p}$ is an another  standard orthogonal basis of  $\mathbb{R}^{p}$ and  $\{\boldsymbol{v}_{k}\}_{k=1}^{q}$ is an another standard orthogonal basis of $\mathbb{R}^{q}$, then the matrix $B$ can be expressed as
\begin{center}
$B=\sum\limits_{j=1}^{p}\sum\limits_{k=1}^{q} \Theta_{j,k}\boldsymbol{u}_{j}\boldsymbol{v}^{\top}_{k}$,
\end{center}
which means there exists a matrix $\Theta\in\mathbb{R}^{p\times q}$ such that
\begin{center}
$B=U\Theta V^{\top}=\sum\limits_{j=1}^{p}\sum\limits_{k=1}^{q}\Theta_{j,k}\boldsymbol{u}_{j}\boldsymbol{v}^{\top}_{k}$.
\end{center}
In this case, $U=\left(\boldsymbol{u}_{1};\boldsymbol{u}_{2};\cdots;\boldsymbol{u}_{p}\right)\in\mathbb{O}^{p}$ and $V=\left(\boldsymbol{v}_{1};\boldsymbol{v}_{2};\cdots;\boldsymbol{v}_{p}\right)\in\mathbb{O}^{q}$.

Specially, if matrices $U$ and $V$ are singular matrices of $B$, then
\begin{center}
$B=\sum\limits_{j=1}^{r}\sigma_{j}(B)\boldsymbol{u}_{j}\boldsymbol{v}^{\top}_{j}$,
\end{center}
where $r\leq\min\{p,q\}$ is the rank of $B$, $\sigma_{j}(B)$ are singular values of $B$. Based on these results, for any matrix $B$ with rank being $r$, its linear combination number can be $r$.

\subsection{Safe Subspace Screening Rule}
Here, we assume $\mathcal{X}\in\mathbb{R}^{n\times pq}$ has full row rank, which is easily satisfied when $n\leq pq$. Given two orthogonal matrixes $U=(\textbf{u}_{1},\textbf{u}_{2},\cdots,\textbf{u}_{p})\in\mathbb{R}^{p\times p}$ with $\left\{\textbf{u}_{1},\textbf{u}_{2},\cdots,\textbf{u}_{p}\right\}$ being an orthogonal basis of $\mathbb{R}^{p\times p}$ and $V=(\textbf{v}_{1},\textbf{v}_{2},\cdots,\textbf{v}_{q})\in\mathbb{R}^{q\times q}$ with $\left\{\textbf{v}_{1},\textbf{v}_{2},\cdots,\textbf{v}_{q}\right\}$ being an orthogonal basis of $\mathbb{R}^{q\times q}$. Then, for any matrix $B\in\mathbb{R}^{p\times q}$, there exists a matrix $\Theta\in\mathbb{R}^{p\times q}$ such that
\begin{center}
$B=U\Theta V^{\top}=\sum\limits_{j=1}^{p}\sum\limits_{k=1}^{q}\Theta_{j,k}\textbf{u}_{j}\textbf{v}^{\top}_{k}$, $\|W_{1}BW_{2}\|_{*}=\|W_{1}U\Theta V^{\top}W_{2}\|_{*}$.
\end{center}
Therefore, the solution of (\ref{eq1}) can be obtained by solving the following model.
\begin{eqnarray}\label{eq3}
\underset{\Theta\in \mathbb{R}^{p\times q}}\min\left\{\frac{1}{2n}\sum\limits_{i=1}^{n}\left(y_{i}-
\left\langle X_{i},U\Theta V^{\top}\right\rangle\right)^{2}+\lambda\|W_{1}U\Theta V^{\top}W_{2}\|_{*}\right\},
\end{eqnarray}
Denote $\hat{\Theta}(\lambda)$ as the solution of this model under $\lambda$. Then the solution $\hat{B}(\lambda)=U\hat{\Theta}(\lambda)V^{\top}$.  According to this new expression of  the solution $\hat{B}(\lambda)$, there is a basic result that can be used to simplify the calculation of $\hat{\Theta}(\lambda)$, which is $\hat{\Theta}_{j,k}(\lambda)=0$ if
\begin{center}
$\left\langle \hat{B}(\lambda), \textbf{u}_{j}\textbf{v}^{\top}_{k}\right\rangle=\left\langle \rm{vec}(\hat{\textit{B}}(\lambda)),\rm{vec}(\textbf{u}_{\textit{j}}\textbf{v}^{\top}_{\textit{k}})\right\rangle=0$.
\end{center}
Suppose there are two indexes $I_{1}$ and $I_{2}$ such that $I_{1}=\left\{j:\hat{\Theta}_{j,:}(\lambda)=0\right\}$ and $I_{2}=\left\{k:\hat{\Theta}_{:,k}(\lambda)=0\right\}$. Then, the problem (\ref{eq3}) can be solved by the following reduced size problem.
\begin{eqnarray*}
\underset{\Theta\in \mathbb{R}^{|I^{c}_{1}|\times |I^{c}_{2}|}}\min\left\{\frac{1}{2n}\sum\limits_{i=1}^{n}\left(y_{i}-
\left\langle X_{i},U_{:,I^{c}_{1}}\Theta V_{:,I^{c}_{2}}^{\top}\right\rangle\right)^{2}+\lambda\|W_{1}U_{:I^{c}_{1}}\Theta V_{:I^{c}_{2}}^{\top}W_{2}\|_{*}\right\}.
\end{eqnarray*}

Next, we estimate the solution $\hat{B}(\lambda)$ and eliminate zero elements of $\hat{\Theta}(\lambda)$, with the help of duality theory.
\begin{Theorem}\label{theorem4.1}
For any $\lambda_{0}\in[0,\lambda_{max})$ and $\lambda\in(\lambda_{0},\lambda_{max}]$. Suppose that the solution $\hat{\theta}(\lambda_{0})$ of (\ref{eq4}) is given. Let $j\in\{1,2,\cdots,p\}$, $k\in\{1,2,\cdots,q\}$ and $W_{j,k}(\lambda)=\max\left\{P^{(1)}_{j,k}(\lambda),P^{(2)}_{j,k}(\lambda)\right\}$, where

\noindent $P^{(1)}_{j,k}(\lambda)=\left\langle\mathcal{X}^{\top}(\mathcal{X}\mathcal{X}^{\top})^{-1}\boldsymbol{y}, \rm{vec}(\textbf{\textit{u}}_{\textit{j}}\textbf{\textit{v}}^{\top}_{\textit{k}})\right\rangle+f_{opt}$,
\begin{equation*}
f_{opt}=
\begin{cases}
\langle c,\gamma\rangle+\frac{\|\gamma\|^{2}}{2\nu}+\left(\frac{b-\langle c,\alpha\rangle}{\|\alpha\|^{2}}+\frac{\langle \gamma,\alpha\rangle}{2\nu\|\alpha\|^{2}}\right)\langle \alpha,\gamma\rangle, & \langle\gamma,\alpha\rangle<2\nu(b-\langle c,\alpha\rangle),\\
\langle c,\gamma\rangle+\eta\|\gamma\|, &\rm{else}.
\end{cases}
\end{equation*}

\noindent $P^{(2)}_{j,k}(\lambda)=-\left\langle\mathcal{X}^{\top}(\mathcal{X}\mathcal{X}^{\top})^{-1}\boldsymbol{y}, \rm{vec}(\textbf{\textit{u}}_{\textit{j}}\textbf{\textit{v}}^{\top}_{\textit{k}})\right\rangle+\tilde{f}_{opt},$
\begin{equation*}
\tilde{f}_{opt}=
\begin{cases}
-\langle c,\gamma\rangle+\frac{\|\gamma\|^{2}}{2\nu}-\left(\frac{b-\langle c,\alpha\rangle}{\|\alpha\|^{2}}+\frac{\langle \gamma,\alpha\rangle}{2\nu\|\alpha\|^{2}}\right)\langle \alpha,\gamma\rangle, & -\langle\gamma,\alpha\rangle<2\nu(b-\langle c,\alpha\rangle),\\
-\langle c,\gamma\rangle+\eta\|\gamma\|, &\rm{else}.
\end{cases}
\end{equation*}
Then,  $\hat{\Theta}_{j,:}(\lambda)=0$ if $\|W_{j,:}(\lambda)\|_{\infty}=0$ and $\hat{\Theta}_{:,k}(\lambda)=0$ if $\|W_{:,k}(\lambda)\|_{\infty}=0$.
\end{Theorem}
\begin{proof}

By simple calculations, we get the dual problem of (\ref{eq1}), which equals to the following problem as
\begin{equation}\label{eq4}
\begin{split}
&\underset{\theta \in \mathbb{R}^{n}}\min ~\frac{n\lambda^{2}}{2}\left\|\theta+\frac{\boldsymbol{y}}{n\lambda}\right\|^{2}\\
&s.t.\quad \left\|W^{-1}_{1}\left(\sum\limits_{i=1}^{n} \theta_{i}X_{i}\right)W^{-1}_{2}\right\|_{2}\leq1.\\
\end{split}
\end{equation}
Let $\hat{\theta}(\lambda)$ be the solution of the dual problem. The KKT system of (\ref{eq1}) and (\ref{eq4}) is
\begin{eqnarray*}
\begin{cases}
0\in W^{-1}_{1}\left(\sum\limits_{i=1}^{n} \theta_{i}X_{i}\right)W^{-1}_{2}+W_{1}\partial_{W_{1}BW_{2}}\|W_{1}BW_{2}\|_{*}W_{2},\\
\boldsymbol{y}-\mathcal{X}\rm{vec}(\textit{B})+\textit{n}\lambda\theta=0.
\end{cases}
\end{eqnarray*}
Based on the result in Rockafellar \cite{R70},  solutions of (\ref{eq1}) and (\ref{eq4}) satisfy the KKT system, which leads to
\begin{center}
$\boldsymbol{y}-\mathcal{X}\rm{vec}(\hat{\textit{B}}(\lambda))+\textit{n}\lambda\hat{\theta}(\lambda)=0$.
\end{center}
So,  solutions of (\ref{eq1}) and (\ref{eq4}) satisfy vec$(\hat{B}(\lambda))=\mathcal{X}^{\top}(\mathcal{X}\mathcal{X}^{\top})^{-1}(\boldsymbol{y}+n\lambda\hat{\theta}(\lambda))$. For any $\lambda>\lambda_{0}>0$, according to the optimality condition of the dual problem (\ref{eq4}), the solution $\hat{\theta}(\lambda)\in\Omega$ with $\Omega$ defined as
\begin{center}
$\Omega=\left\{\theta: \left\langle \hat{\theta}(\lambda_{0})+\frac{\boldsymbol{y}}{n\lambda_{0}},\theta-\hat{\theta}(\lambda_{0})\right\rangle\geq0,
\left\langle \theta+\frac{\boldsymbol{y}}{n\lambda},\hat{\theta}(\lambda_{0})-\theta\right\rangle\geq0\right\}$.
\end{center}
Then $\hat{\Theta}_{j,k}(\lambda)=0$ if
\begin{center}
$\underset{\theta\in\Omega}\max\Big|\left\langle \mathcal{X}^{\top}(\mathcal{X}\mathcal{X}^{\top})^{-1}(\boldsymbol{y}+n\lambda\theta),\rm{vec}(\textbf{\textit{u}}_{\textit{j}}
\textbf{\textit{v}}^{\top}_{\textit{k}})\right\rangle\Big|$
$=\max\left\{P^{(1)}_{j,k}(\lambda), P^{(2)}_{j,k}(\lambda)\right\}=0$,
\end{center}
where
\begin{equation*}
\begin{split}
P^{(1)}_{j,k}(\lambda)&=\underset{\theta\in\Omega}\max\left\langle \mathcal{X}^{\top}(\mathcal{X}\mathcal{X}^{\top})^{-1}(\boldsymbol{y}+n\lambda\theta),\rm{vec}(\textbf{\textit{u}}_{\textit{j}}\textbf{\textit{v}}^{\top}_{\textit{k}})\right\rangle\\
&=\left\langle\mathcal{X}^{\top}(\mathcal{X}\mathcal{X}^{\top})^{-1}\boldsymbol{y}, \rm{vec}(\textbf{\textit{u}}_{\textit{j}}\textbf{\textit{v}}^{\top}_{\textit{k}})\right\rangle+\underset{\theta\in\Omega}\max\left\langle n\lambda\mathcal{X}^{\top}(\mathcal{X}\mathcal{X}^{\top})^{-1}\theta,\rm{vec}(\textbf{\textit{u}}_{\textit{j}}\textbf{\textit{v}}^{\top}_{\textit{k}})\right\rangle
\end{split}
\end{equation*}
and
\begin{equation*}
\begin{split}
P^{(2)}_{j,k}(\lambda)&=\underset{\theta\in\Omega}\max\left\langle -\mathcal{X}^{\top}(\mathcal{X}\mathcal{X}^{\top})^{-1}(\boldsymbol{y}+n\lambda\theta),
\rm{vec}(\textbf{\textit{u}}_{\textit{j}}\textbf{\textit{v}}^{\top}_{\textit{k}})\right\rangle\\
&=\langle-\mathcal{X}^{\top}(\mathcal{X}\mathcal{X}^{\top})^{-1}\boldsymbol{y}, \rm{vec}(\textbf{\textit{u}}_{\textit{j}}\textbf{\textit{v}}^{\top}_{\textit{k}})\rangle+\underset{\theta\in\Omega}\max\left\langle -\textit{n}\lambda\mathcal{X}^{\top}(\mathcal{X}\mathcal{X}^{\top})^{-1}\theta,\rm{vec}(\textbf{\textit{u}}_{\textit{j}}\textbf{\textit{v}}^{\top}_{\textit{k}})\right\rangle.
\end{split}
\end{equation*}
Now, we give detailed results of these two values. In order to be understandable, we introduce some new  notations as follows.
\begin{equation}\label{n1}
\begin{aligned}
&\gamma=n\lambda(\mathcal{X}\mathcal{X}^{\top})^{-1}\mathcal{X}\rm{vec}(\textbf{\textit{u}}_{\textit{j}}\textbf{\textit{v}}^{\top}_{\textit{k}}), \\ &\alpha=\hat{\theta}(\lambda_{0})+\frac{\boldsymbol{y}}{n\lambda_{0}},
\quad b=\left\langle\hat{\theta}(\lambda_{0}),\alpha\right\rangle, \\
&c=\frac{1}{2}\left(\hat{\theta}(\lambda_{0})-\frac{\boldsymbol{y}}{n\lambda}\right), \quad \eta^{2}=\frac{1}{4}\left\|\hat{\theta}(\lambda_{0})+\frac{\boldsymbol{y}}{n\lambda}\right\|^{2},\\
&2\nu=\sqrt{\frac{\|\gamma\|^{2}\|\alpha\|^{2}-\langle\gamma,\alpha\rangle^{2}}{\|\alpha\|^{2}\eta^{2}-(b-\langle c,\alpha\rangle)^{2}}}.
\end{aligned}
\end{equation}

Here, we give the detailed process to calculate the value of $P^{(1)}_{j,k}(\lambda)$, which is decided by the value of the following problem.
\begin{equation*}
\begin{split}
&\underset{\theta \in \mathbb{R}^{n}}\max ~\left\langle n\lambda\mathcal{X}^{\top}(\mathcal{X}\mathcal{X}^{\top})^{-1}\theta,
\rm{vec}(\textbf{\textit{u}}_{\textit{j}}\textbf{\textit{v}}^{\top}_{\textit{k}})\right\rangle\\
&s.t.\quad \left\langle \hat{\theta}(\lambda_{0})+\frac{\boldsymbol{y}}{n\lambda_{0}},\theta-\hat{\theta}(\lambda_{0})\right\rangle\geq0,\\
&\quad \quad~~\left\langle \theta+\frac{\boldsymbol{y}}{n\lambda},\hat{\theta}(\lambda_{0})-\theta\right\rangle\geq0.
\end{split}
\end{equation*}
This problem is equivalent to
\begin{equation}\label{eq5}
\begin{split}
&\underset{\theta \in \mathbb{R}^{n}}\max ~\left\langle n\lambda(\mathcal{X}\mathcal{X}^{\top})^{-1}\mathcal{X}\rm{vec}(\textbf{\textit{u}}_{\textit{j}}\textbf{\textit{v}}^{\top}_{\textit{k}}),\theta\right\rangle\\
&s.t.\quad \left\langle \theta,\hat{\theta}(\lambda_{0})+\frac{\boldsymbol{y}}{n\lambda_{0}}\right\rangle
-\left\langle\hat{\theta}(\lambda_{0}),\hat{\theta}(\lambda_{0})+\frac{\boldsymbol{y}}{n\lambda_{0}}\right\rangle\geq0,\\
&\quad ~~\quad  \frac{1}{4}\left\|\hat{\theta}(\lambda_{0})+\frac{\boldsymbol{y}}{n\lambda}\right\|^{2}-\left\|\theta-\frac{1}{2}\left(\hat{\theta}(\lambda_{0})-\frac{\boldsymbol{y}}{n\lambda}\right)\right\|^{2}
\geq 0.
\end{split}
\end{equation}
Based on notations in the equation (\ref{n1}), the problem (\ref{eq5}) is equivalent to
\begin{equation}\label{eq6}
\begin{split}
f_{opt}&=\underset{\theta \in \mathbb{R}^{n}}\max ~\langle \gamma,\theta\rangle\\
&s.t.\quad \langle \theta,\alpha\rangle-b\geq 0,\\
&\quad ~~\quad  \eta^{2}-\left\|\theta-c\right\|^{2}\geq0 .
\end{split}
\end{equation}
By introducing Lagrange multipliers $\mu\geq0$ and $\nu\geq0$, Lagrangian function of (\ref{eq6}) is
\begin{center}
$\mathcal{L}(\theta;\mu,\nu)=\langle \gamma,\theta\rangle+\mu(\langle \theta,\alpha\rangle-b)+\nu(\eta^{2}-\left\|\theta-c\right\|^{2})$.
\end{center}
To obtain the dual problem of (\ref{eq6}), we need to compute the value of the following problem.
\begin{equation*}
\begin{aligned}
\underset{\theta}\max\mathcal{L}(\theta;\mu,\nu)
&= \underset{\theta}\max\left\{\langle \gamma+\mu\alpha,\theta\rangle-\nu\left\|\theta-c\right\|^{2}\right\}-\mu b+\nu\eta^{2}\\
&= \begin{cases}
\frac{\|\gamma+\mu\alpha\|^{2}}{4\nu}+\langle c,\gamma+\mu\alpha\rangle-\mu b+\nu \eta^{2},& \nu>0,\\
+\infty, &\nu=0.
\end{cases}
\end{aligned}
\end{equation*}
Therefore, the dual problem of (\ref{eq6}) is
\begin{equation*}
\begin{aligned}
\underset{\mu\geq0,\nu\geq0}\min\underset{\theta}\max\mathcal{L}(\theta;\mu,\nu)=\underset{\mu\geq0,\nu>0}\min ~\frac{\|\gamma+\mu\alpha\|^{2}}{4\nu}+\langle c,\gamma+\mu\alpha\rangle-\mu b+\nu \eta^{2}.
\end{aligned}
\end{equation*}
The KKT system is
\begin{eqnarray*}
\begin{cases}
\theta=c+\frac{1}{2\nu}(\gamma+\mu\alpha),\\
\mu(\langle \theta,\alpha\rangle-b)=0,\mu\geq0,\langle \theta,\alpha\rangle-b\geq 0,\\
\eta^{2}-\left\|\theta-c\right\|^{2}=0,\nu>0.
\end{cases}
\end{eqnarray*}

Next, we solve the problem (\ref{eq6}) based on the KKT system. There are two cases of solutions of the KKT system.

\underline{Cases 1:} $\mu>0$. In this case, $\langle \theta,\alpha\rangle=b$. Replacing $\theta=c+\frac{1}{2\nu}(\gamma+\mu\alpha)$ from the KKT system into this equality, we know that $\langle c,\alpha\rangle+\frac{1}{2\nu}\langle \gamma+\mu\alpha,\alpha\rangle=b$, which leads to
\begin{center}
$\mu=\frac{2\nu(b-\langle c,\alpha\rangle)-\langle\gamma,\alpha\rangle}{\|\alpha\|^{2}}$.
\end{center}
To ensure that $\mu>0$, $\langle\gamma,\alpha\rangle<2\nu(b-\langle c,\alpha\rangle)$ needs to be assumed.
Replacing $\theta=c+\frac{1}{2\nu}(\gamma+\mu\alpha)$ into $\eta^{2}-\left\|\theta-c\right\|^{2}=0$, the closed-form formula of $2\nu$ is
\begin{eqnarray}\label{eq7}
2\nu=\sqrt{\frac{\|\gamma\|^{2}\|\alpha\|^{2}-\langle\gamma,\alpha\rangle^{2}}{\|\alpha\|^{2}\eta^{2}-(b-\langle c,\alpha\rangle)^{2}}}.
\end{eqnarray}
Based on Cauchy inequality,  $\|\gamma\|^{2}\|\alpha\|^{2}>\langle\gamma,\alpha\rangle^{2}$. Due to the notations of $\alpha$, $\eta$, $b$ and $c$, we know that $\|\alpha\|^{2}\eta^{2}-(b-\langle c,\alpha\rangle)^{2}>0$. So, the closed-form of $2\nu$ in (\ref{eq7}) exists and the solution of  $(\ref{eq6})$ can be obtained, which is
 \begin{center}
$\theta=c+\frac{b-\langle c,\alpha\rangle}{\|\alpha\|^{2}}\alpha-\frac{\langle \gamma,\alpha\rangle}{2\nu\|\alpha\|^{2}}\alpha+\frac{\gamma}{2\nu}$.
\end{center}

\underline{Cases 2:} $\mu=0$. In this case, $\theta=c+\frac{\gamma}{2\nu}$. Replacing this expression into $\eta^{2}-\left\|\theta-c\right\|^{2}=0$, the closed-form formula of $\nu$ is
$\nu=\frac{\|\gamma\|}{2\eta}$, which means
$$\theta=c+\frac{\eta}{\|\gamma\|}\gamma.$$

Combing these results, we get that
\begin{equation*}
f_{opt}=
\begin{cases}
\langle c,\gamma\rangle+\frac{\|\gamma\|^{2}}{2\nu}+\left(\frac{b-\langle c,\alpha\rangle}{\|\alpha\|^{2}}+\frac{\langle \gamma,\alpha\rangle}{2\nu\|\alpha\|^{2}}\right)\langle \alpha,\gamma\rangle, & \langle\gamma,\alpha\rangle<2\nu(b-\langle c,\alpha\rangle),\\
\langle c,\gamma\rangle+\eta\|\gamma\|, &\rm{else},
\end{cases}
\end{equation*}
where $2\nu$ is defined in the equation (\ref{eq7}). Hence, $$P^{(1)}_{j,k}(\lambda)=\left\langle(\mathcal{X}\mathcal{X}^{\top})^{-1}\mathcal{X}^{\top}\boldsymbol{y}, \rm{vec}(\textbf{\textit{u}}_{\textit{j}}\textbf{\textit{v}}^{\top}_{\textit{k}})\right\rangle+f_{opt}.$$
In the same way, the closed form of $P^{(2)}_{j,k}(\lambda)$ can be obtained. For simplicity of the paper, we just present the result of $P^{(2)}_{j,k}(\lambda)$.
$$P^{(2)}_{j,k}(\lambda)=-\left\langle(\mathcal{X}\mathcal{X}^{\top})^{-1}\mathcal{X}^{\top}\boldsymbol{y}, \rm{vec}(\textbf{\textit{u}}_{\textit{j}}\textbf{\textit{v}}^{\top}_{\textit{k}})\right\rangle+\tilde{f}_{opt},$$
where
\begin{equation*}
\tilde{f}_{opt}=
\begin{cases}
-\langle c,\gamma\rangle+\frac{\|\gamma\|^{2}}{2\nu}-\left(\frac{b-\langle c,\alpha\rangle}{\|\alpha\|^{2}}+\frac{\langle \gamma,\alpha\rangle}{2\nu\|\alpha\|^{2}}\right)\langle \alpha,\gamma\rangle, & -\langle\gamma,\alpha\rangle<2\nu(b-\langle c,\alpha\rangle),\\
-\langle c,\gamma\rangle+\eta\|\gamma\|, &\rm{else}.
\end{cases}
\end{equation*}
\end{proof}

Based on this theorem, the inactive rows and columns in the solution $\hat{\Theta}(\lambda)$ should be easily identified, which will reduce the dimension and accelerate the computation of the solution.

So far, an unspecified issue is about how to compute the solution  of (\ref{eq1}), under different values of $\lambda$. This is a fundamental step to evaluate the efficiency  safe subspace screening rule. In this section, we solve the problem (\ref{eq1}) via ADMM (Alternating Direction Method of Multipliers) in Algorithm 1.  To illustrate the computation effect of our screening rule, we embed the result in Theorem 3.1 in ADMM and present a accelerate version in Algorithm 2.
\begin{table}[htbp]
\centering
\begin{tabular}{l}
\hline
\textbf{Algorithm 1}: ADMM for solving (\ref{eq1}).                                                                                                                                                                                                                                                                    \\ \hline
1:\quad  \textbf{Input}: a tuning parameter $0<\lambda<\lambda_{max}$.                                                                                                                                              \\
2:\quad  \textbf{Output}: $\hat{B}(\lambda)=B^{k}$.                                                                                                                                                                                                    \\
3:\quad  \textbf{Initialization}: $k=0$, $B^{k}=0$, $\alpha^{k}=0$, $C^{k}=0$, $\theta^{k}=0$, $D^{k}=0$, $\tau\in\left(0,\frac{\sqrt{5}+1}{2}\right)$, $\sigma>0$. \\
4:\quad\quad~4.1.\quad vec$(B^{k+1})=\underset{\rm{vec}(\textit{B})}{\arg\min} \mathcal{L}_{\sigma}(B,\alpha^{k},C^{k};\theta^{k},D^{k})$;\\
\quad\quad\quad 4.2.\quad $\alpha^{k+1}=\underset{\alpha}{\arg\min} \mathcal{L}_{\sigma}(B^{k+1},\alpha,C^{k};\theta^{k},D^{k})$;\\
\quad\quad\quad 4.3.\quad  $C^{k+1}=\underset{C}{\arg\min} \mathcal{L}_{\sigma}(B^{k+1},\alpha^{k},C;\theta^{k},D^{k})$;\\
\quad\quad\quad 4.4.\quad $\theta^{k+1}=\theta^{k}-\tau\sigma\left[\boldsymbol{y}-\mathcal{X}\rm{vec}(\textit{B}^{\textit{k}+1})-\alpha^{\textit{k}+1}\right]$;\\
\quad\quad\quad 4.5.\quad $D^{k+1}=D^{k}-\tau\sigma(W_{1}B^{k+1}W_{2}-D_{k})$;\\
5:\quad If the termination conditions are not met, $k=k+1$ and go to step 4.
\\ \hline
\end{tabular}
\end{table}

First, we transform (\ref{eq1}) as a constrained problem.
\begin{equation*}
\begin{aligned}
&\underset{B\in\mathbb{R}^{p\times q},\alpha\in\mathbb{R}^{n},C\in\mathbb{R}^{p\times q}}\min
\frac{1}{2n}\|\alpha\|^{2}+\lambda\|C\|_{*}\\
& s.t.\quad \boldsymbol{y}-\mathcal{X}\rm{vec}(\textit{B})-\alpha=0, W_{1}BW_{2}-C=0.
\end{aligned}
\end{equation*}
Then, the augmented Lagrangian function of this constrained problem is
\begin{align*}
&\mathcal{L}_{\sigma}(B,\alpha,C;\theta,D)\\
&=\frac{1}{2n}\|\alpha\|^{2}+\lambda\|C\|_{*}-\left\langle \theta, \boldsymbol{y}-\mathcal{X}\rm{vec}(\textit{B})-\alpha\right\rangle-\left\langle D, W_{1}BW_{2}-C\right\rangle+\frac{\sigma}{2}\|\boldsymbol{y}-\mathcal{X}\rm{vec}(\textit{B})-\alpha\|^{2}+
\frac{\sigma}{2}\|\textit{W}_{1}\textit{BW}_{2}-\textit{C}\|^{2}_{\textit{F}}.
\end{align*}
Then, ADMM for solving (\ref{eq1}) is reviewed in Algorithm 1. Here, we compute details in Algorithm 1.
\begin{align*}
\rm{vec}(\textit{B}^{\textit{k}+1})&=\underset{\rm{vec}(\textit{B})}{\arg\min} \mathcal{L}_{\sigma}(B,\alpha^{k},C^{k};\theta^{k},D^{k})\\
&=\underset {\rm{vec}(\textit{B})}{\arg\min} \left\{\langle \mathcal{X}^{\top}\theta^{k}
-(W_{2}\bigotimes W_{1})\rm{vec}(\textit{D}^{\textit{k}}),\rm{vec}(\textit{B})\rangle+\frac{\sigma}{2}\|\boldsymbol{y}-\mathcal{X}\rm{vec}(\textit{B})-\alpha^{\textit{k}}\|^{2}+
\frac{\sigma}{2}\|(W_{2}\bigotimes W_{1})\rm{vec}(\textit{B})-\rm{vec}(\textit{C}^{\textit{k}})\|^{2}\right\}\\
&=\frac{1}{\sigma}\left(\mathcal{X}^{\top}\mathcal{X}+W^{2}_{2}\bigotimes W^{2}_{1}\right)^{-1}
\left[\left(W_{2}\bigotimes W_{1}\right)\rm{vec}(\textit{D}^{\textit{k}}+\sigma\textit{C}^{\textit{k}})+\sigma\mathcal{X}^{\top}\boldsymbol{y}-\mathcal{X}^{\top}\theta^{k}
-\sigma\mathcal{X}^{\top}\alpha^{k}\right].\\
\alpha^{k+1}&=\underset {\alpha}{\arg\min} \mathcal{L}_{\sigma}(B^{k+1},\alpha,C^{k};\theta^{k},D^{k})\\
&=\underset {\alpha} {\arg\min}\left\{\frac{1}{2n}\|\alpha\|^{2}+\langle \theta^{k},\alpha\rangle+\frac{\sigma}{2}\|\boldsymbol{y}-\mathcal{X}\rm{vec}(\textit{B}^{\textit{k}+1})-\alpha\|^{2}\right\}\\
&=\left(\frac{1}{n}+\sigma\right)^{-1}
\left[\sigma\boldsymbol{y}-\theta^{k}-\sigma\mathcal{X}\rm{vec}(\textit{B}^{\textit{k}+1})\right].\\
C^{k+1}&=\underset {C}{\arg\min} \mathcal{L}_{\sigma}(B^{k+1},\alpha^{k},C;\theta^{k},D^{k})\\
&=\underset{C}{\arg\min}\left\{\lambda\|C\|_{*}+\langle D^{k},C\rangle+\frac{\sigma}{2}\|W_{1}B^{k+1}W_{2}-C\|_{F}^{2}\right\}\\
&=\rm{prox}_{\frac{\lambda}{\sigma}\|\cdot\|_{*}}\left(\textit{W}_{1}\textit{B}^{\textit{k}+1}\textit{W}_{2}-\frac{\textit{D}_{\textit{k}}}{\sigma}\right).
\end{align*}

Because that problem (1) is convex and its convergence analysis under ADMM for are already built up in \cite{C17}, we omit this result here.

\begin{table}[htbp]
\centering
\begin{tabular}{l}
\hline
\textbf{Algorithm 2}: Sequential safe subspace screening rule to get the solution path of (\ref{eq1}).                                                                                                                                                                                                                                                                    \\ \hline
1:\quad  \textbf{Input}: a sequence of tuning parameters $0<\lambda_{1}<\lambda_{2}<\cdots<\lambda_{K}<\lambda_{max}$.                                                                                                                                              \\
2:\quad  \textbf{Output}: $\hat{B}(\lambda_{1}), \hat{B}(\lambda_{2}), \cdots, \hat{B}(\lambda_{K})$.                                                                                                                                                                                                    \\
3:\quad  \textbf{Initialization}: vec$(\hat{B}(\lambda_{1}))=\mathcal{X}^{\top}(\mathcal{X}\mathcal{X}^{\top})^{-1}\boldsymbol{y}$, $\hat{\theta}(\lambda_{1})=0$ and $\hat{B}(\lambda_{1})=U$Diag$(b_{\lambda_{1}})V^{\top}$. \\
4:\quad  \textbf{for} $m=2,3,\cdots,K$\\
5:\quad \quad 5.1. \quad Let $\lambda_{0}=\lambda_{m-1}$. According to (\ref{n1}), compute\\
\quad\quad\quad\quad\quad\quad \quad\quad\quad\quad\quad $W_{j,k}(\lambda_{m})=\max\left\{P^{(1)}_{j,k}(\lambda_{m}),P^{(2)}_{j,k}(\lambda_{m})\right\}$\\
\quad\quad\quad \quad\quad \quad for any $j=1,2,\cdots,p$ and $k=1,2,\cdots,q$.\\
\quad\quad\quad 5.2.\quad Let $I_{1}=\left\{j:\|W_{j,:}(\lambda_{m})\|_{\infty}=0\right\}$ and $I_{2}=\left\{k:\|W_{:,k}(\lambda_{m})\|_{\infty}=0\right\}$.\\
\quad\quad\quad 5.3.\quad  Solve\\
\quad\quad\quad\quad\quad\quad\quad\quad $\hat{\Theta}(\lambda_{m})=\underset{\Theta\in\mathbb{R}^{|I^{c}_{1}|\times|I^{c}_{2}|}}{\arg\min}
\left\{\frac{1}{2n}\sum\limits_{i=1}^{n}\left(y_{i}-
\left\langle X_{i},U_{:,I^{c}_{1}}\Theta V_{:,I^{c}_{2}}^{\top}\right\rangle\right)^{2}+\lambda_{m}\|W_{1}U_{:,I^{c}_{1}}\Theta V_{:,I^{c}_{2}}^{\top}W_{2}\|_{*}\right\}$.\\
\quad\quad\quad 5.4.\quad Let $\hat{B}(\lambda_{m})=U_{:,I_{1}}\hat{\Theta}(\lambda_{m})V_{:,I_{2}}^{\top}$ and $\hat{\theta}(\lambda_{m})=\frac{1}{n\lambda_{m}}\left[\boldsymbol{y}-\mathcal{X}\rm{vec}\left(\hat{\textit{B}}(\lambda_{\textit{m}})\right)\right]$.\\
\quad\quad\quad 5.5.\quad The singular value decomposition of $\hat{B}(\lambda_{m})$ is $\hat{B}(\lambda_{m})=U$Diag$(b_{\lambda_{m}})V^{\top}$.\\
6:\quad \textbf{end for}
\\ \hline
\end{tabular}
\end{table}
\section{Numerical Experiments}
Here, we evaluate the safe subspace screening rule in Algorithm 2 on some data sets. For each data set, we report the total computation time $T_{f}$ and reduced computation time $T_{s}$, where $T_{f}$ means the total computational time of ADMM in Algorithm 1 under $\lambda_{m}=0.616^{m}\lambda_{max},m=1,2,\cdots,K$  as  in \cite{Z14}, and $T_{s}$ means the computation time of Algorithm 2 under these $\lambda_{m}$. Then, speedup values are computed as the result of $T_{f}/T_{s}$. This quality, as the name illustrated, indicates  the reduced computational time because of the screening rule. The larger the speedup is, the more efficient the screening rule is.
\subsection{Gaussian Distribution Data Sets}
This section considers the simulation data sets in \cite{B08}, where the prediction matrixes follow the normal distribution/Gaussian distribution. The detail process of these simulation data sets are showed as follows. We generate random i.i.d. data $P_{i}\in\mathbb{R}^{p_{1}}$ and $Q_{i}\in\mathbb{R}^{p_{2}}$ with Gaussian distributions and we select a low rank matrix $B\in\mathbb{R}^{p\times q}$ at random. Then,  $y_{i}=\langle P_{i}Q^{\top}_{i},B\rangle + \epsilon_{i},i=1,2,\cdots,n$, where $\epsilon_{i}$ have i.i.d components with normal distributions with zero mean and 0.1 standard variance. Here, we fix the rank of the true matrix $B$ as 2 and $K=20$.

With the simulation data sets, we report the computational results of $T_{f}$, $T_{f}$, $T_{s}$ and speedup values under different $p$, $p$ and $n$ in Table 1. Because that simulation data sets under the fixed dimension are random, we simulate 10 times and report the mean and variance of these results. Based on the results in Table 1, there are some conclusions. Firstly, all $T_{f}$ values are larger than their corresponding $T_{s}$ values, meaning that our safe feature screening rule reduce the computational cost of the model and leading to the speedup values larger than 1.  Secondly, with fixed $p$ and $q$, the speedup values decrease with the sample size $n$ increasing. Lastly, the variances of speedup values is slightly larger than 0, which means our screening rule performs stable under random data sets.
\begin{table}[htbp]
\caption{Computation results under different dimensions of data set and sample sizes, including $T_{f}$, $T_{s}$ and speedup values. This table reports mean values of 10 times simulations and the variance of these results in brackets.}
\begin{tabular}{|c|c|lll|}
\hline
dimension                              & sample size & \multicolumn{1}{c}{$T_{f}$} & \multicolumn{1}{c}{$T_{s}$} & \multicolumn{1}{c|}{speedup}      \\ \hline
\multirow{3}{*}{$p=15$,$q=45$} & $n=30$      & 15.443(0.554)                & 5.094(0.022)                 & 3.034(0.028) \\
                                       & $n=50$      & 17.643(0.150)                & 10.327(0.292)                & 1.712(0.006) \\
                                       & $n=100$     & 20.013(0.556)                & 18.327(1.071)                & 1.095(0.005) \\ \hline
\multirow{3}{*}{$p=25$,$q=45$} & $n=30$      & 28.129(0.039)                & 8.289(0.026)                 & 3.375(0.004) \\
                                       & $n=50$      & 45.234(15.437)               & 12.914(0.659)                & 3.508(0.768) \\
                                       & $n=100$     & 44.484(5.444)                & 34.819(4.334)                & 1.289(0.016) \\ \hline
\multirow{3}{*}{$p=25$,$q=30$} & $n=30$      & 17.958(1.273)                & 7.303(0.116)                 & 2.465(0.039) \\
                                       & $n=50$      & 23.148(4.032)                & 10.719(0.817)                & 2.293(0.025) \\
                                       & $n=100$     & 19.945(1.292)                & 16.819(1.183)                & 1.190(0.005) \\ \hline
\multirow{3}{*}{$p=35$,$q=30$} & $n=30$      & 31.203(0.476)                & 9.114(1.558)                 & 3.482(0.192) \\
                                       & $n=50$      & 32.095(2.155)                & 12.918(1.373)                & 2.497(0.029) \\
                                       & $n=100$     & 36.767(1.852)                & 29.776(0.513)                & 1.236(0.004) \\ \hline
\multirow{3}{*}{$p=40$,$q=20$} & $n=30$      & 21.894(0.107)                & 7.792(0.068)                 & 2.811(0.007) \\
                                       & $n=50$      & 25.749(0.486)                & 16.763(0.443)                & 1.539(0.007) \\
                                       & $n=100$     & 25.517(0.199)                & 22.599(1.390)                & 1.132(0.005) \\ \hline
\end{tabular}
\end{table}

\begin{figure}[htbp]
\centering
\subfigure[device0-9]{
\begin{minipage}[htbp]{0.15\linewidth}
\centering
\includegraphics[width=0.55in]{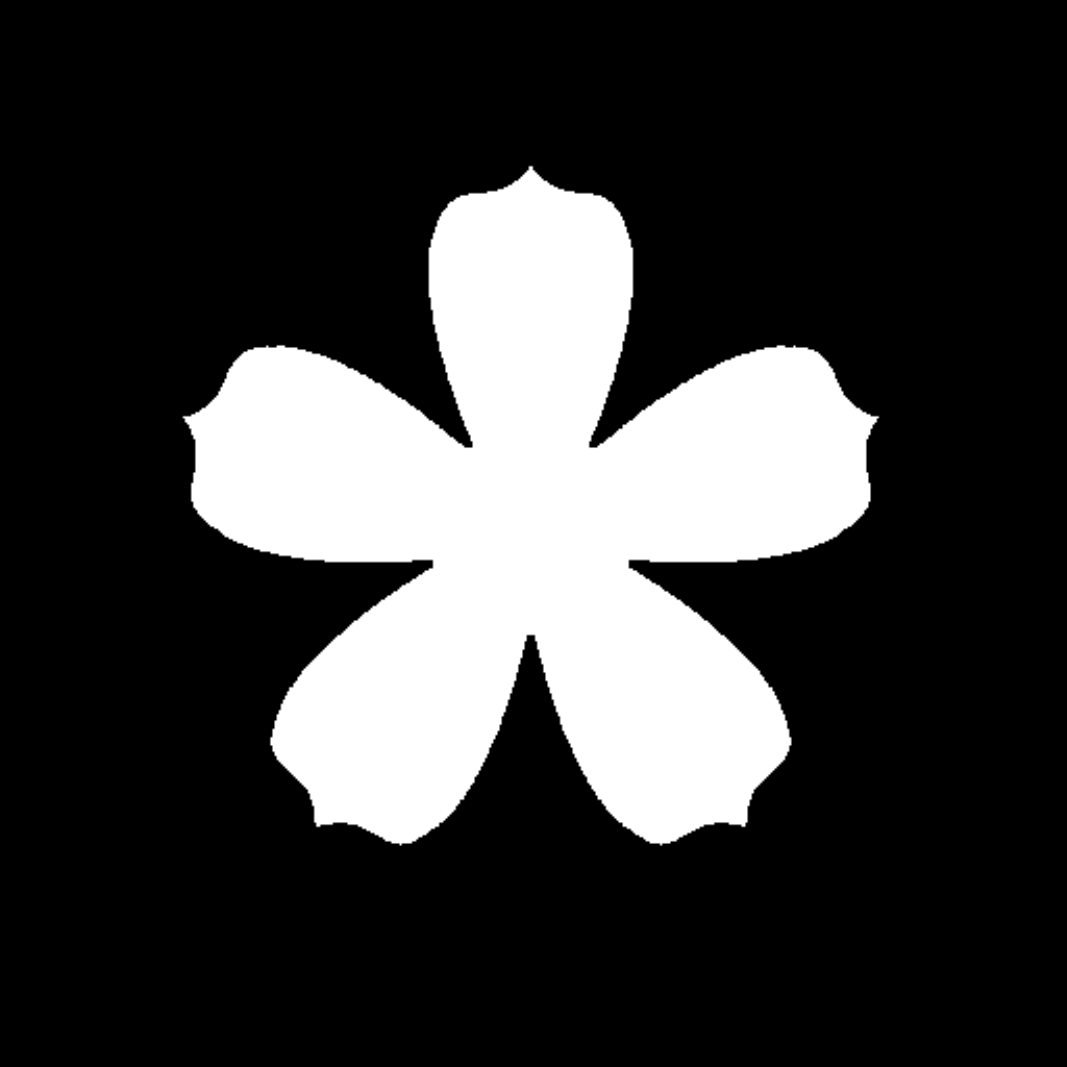}
\end{minipage}
}
\subfigure[device1-5]{
\begin{minipage}[htpb]{0.15\linewidth}
\centering
\includegraphics[width=0.55in]{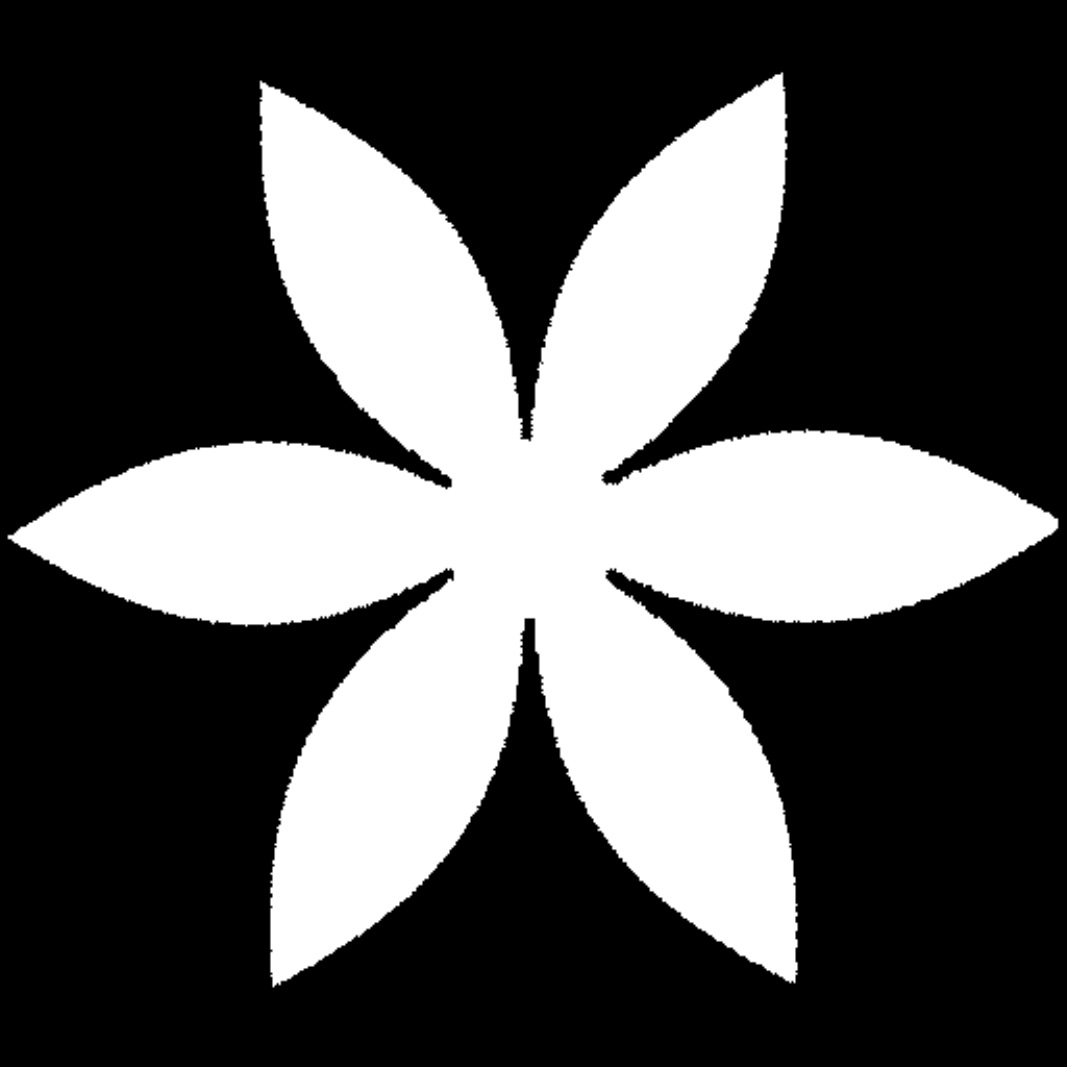}
\end{minipage}
}
\subfigure[device2-1]{
\begin{minipage}[htpb]{0.15\linewidth}
\centering
\includegraphics[width=0.55in]{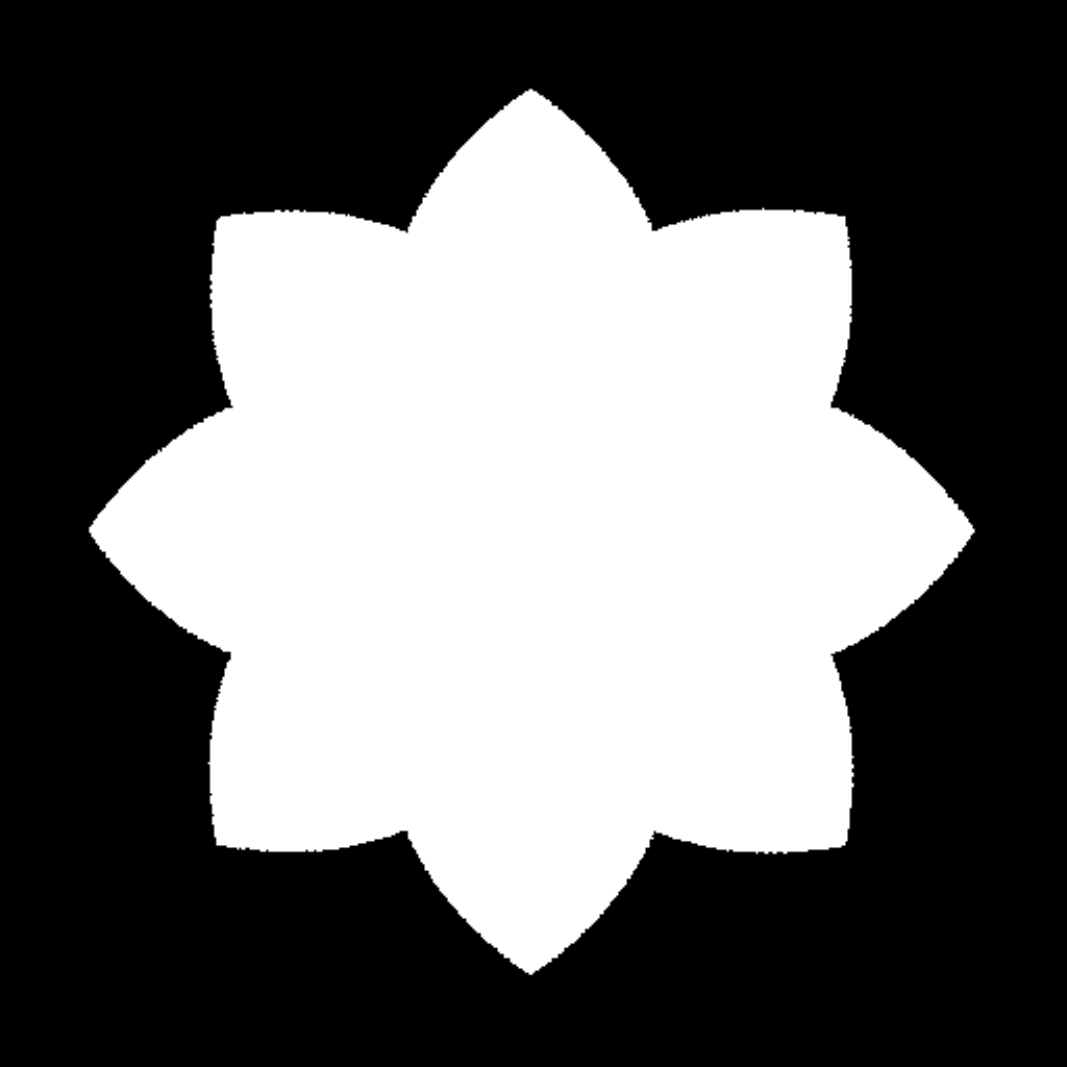}
\end{minipage}
}
\subfigure[device3-11]{
\begin{minipage}[htpb]{0.15\linewidth}
\centering
\includegraphics[width=0.55in]{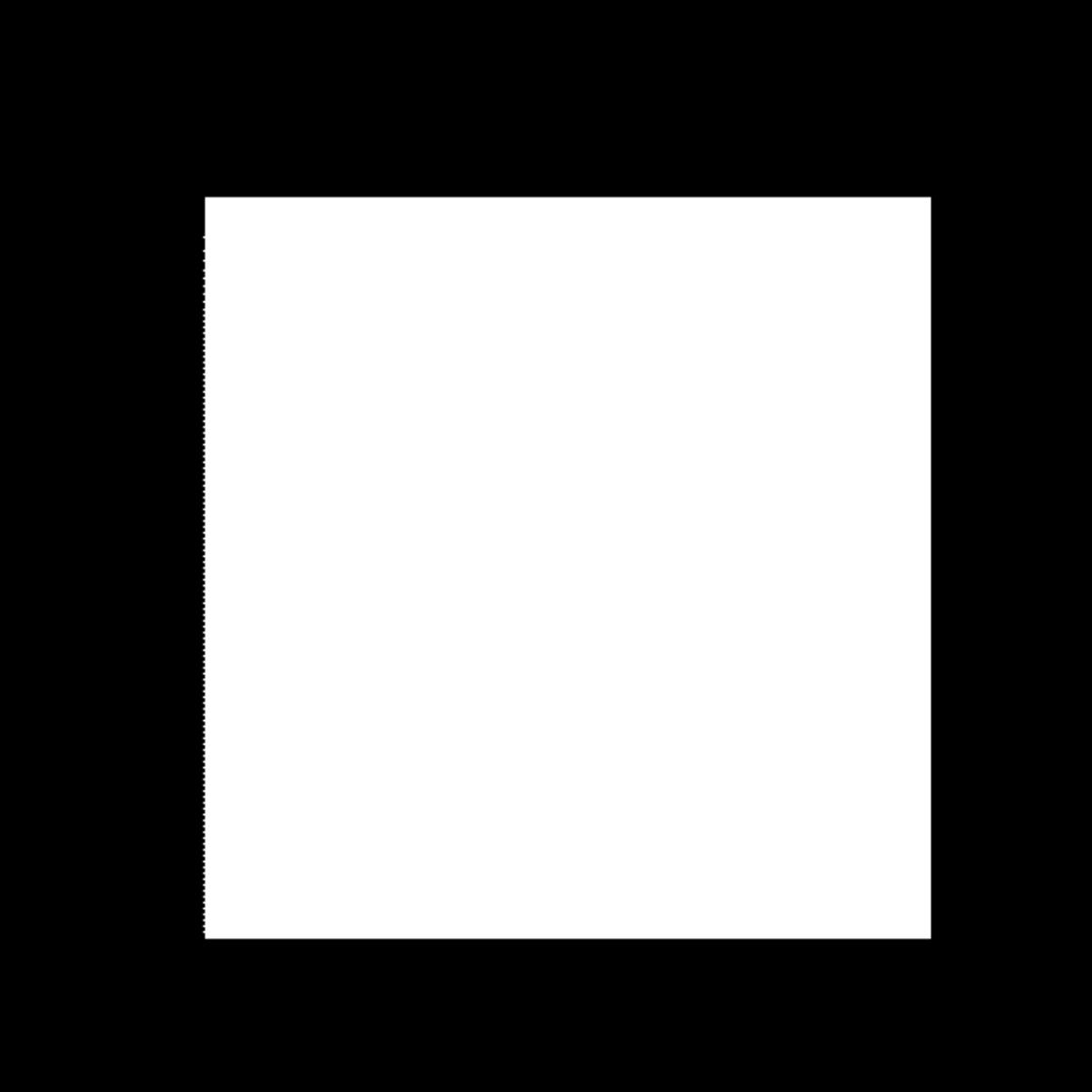}
\end{minipage}
}
\subfigure[device4-20]{
\begin{minipage}[htpb]{0.15\linewidth}
\centering
\includegraphics[width=0.55in]{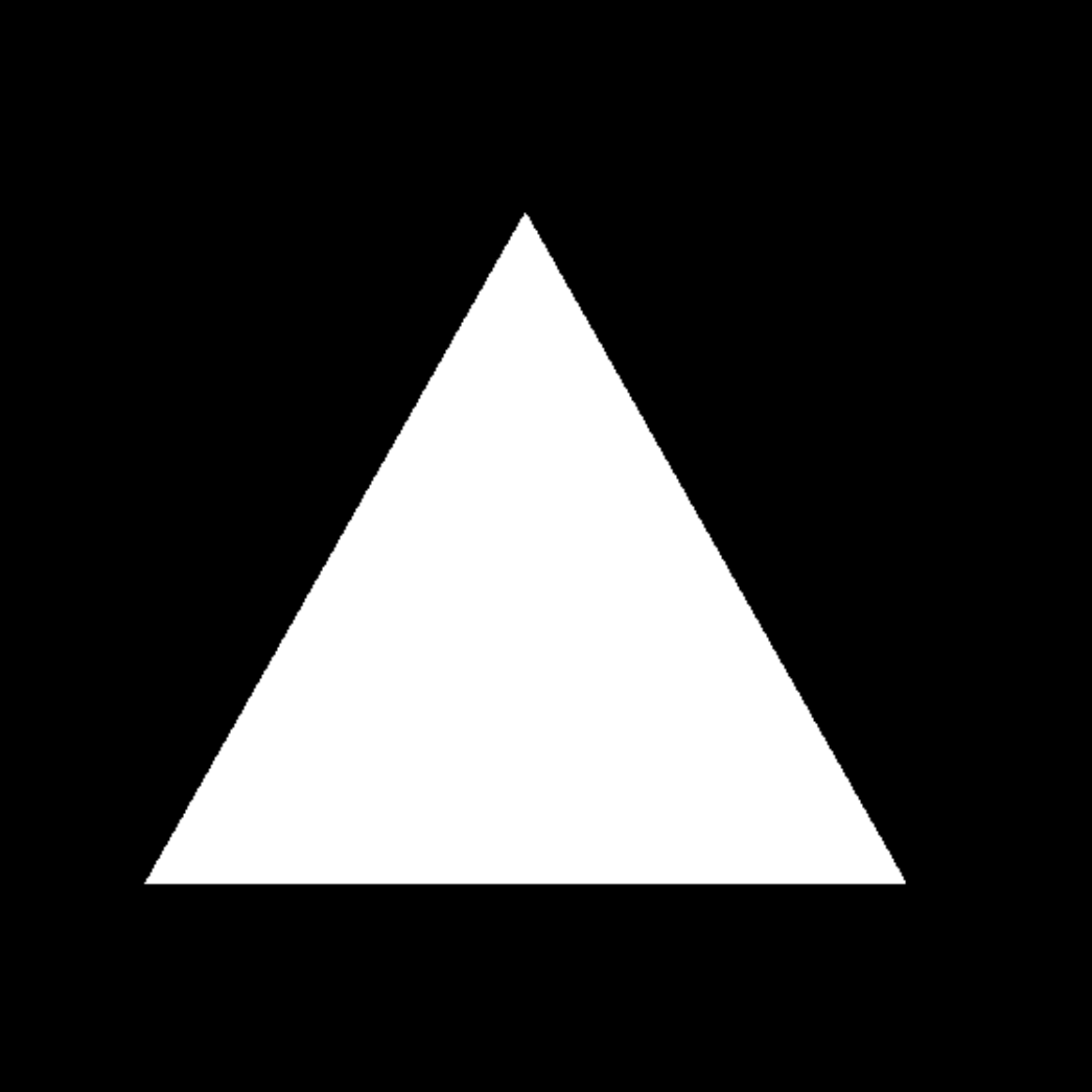}
\end{minipage}
}

\subfigure[device5-20]{
\begin{minipage}[htpb]{0.15\linewidth}
\centering
\includegraphics[width=0.55in]{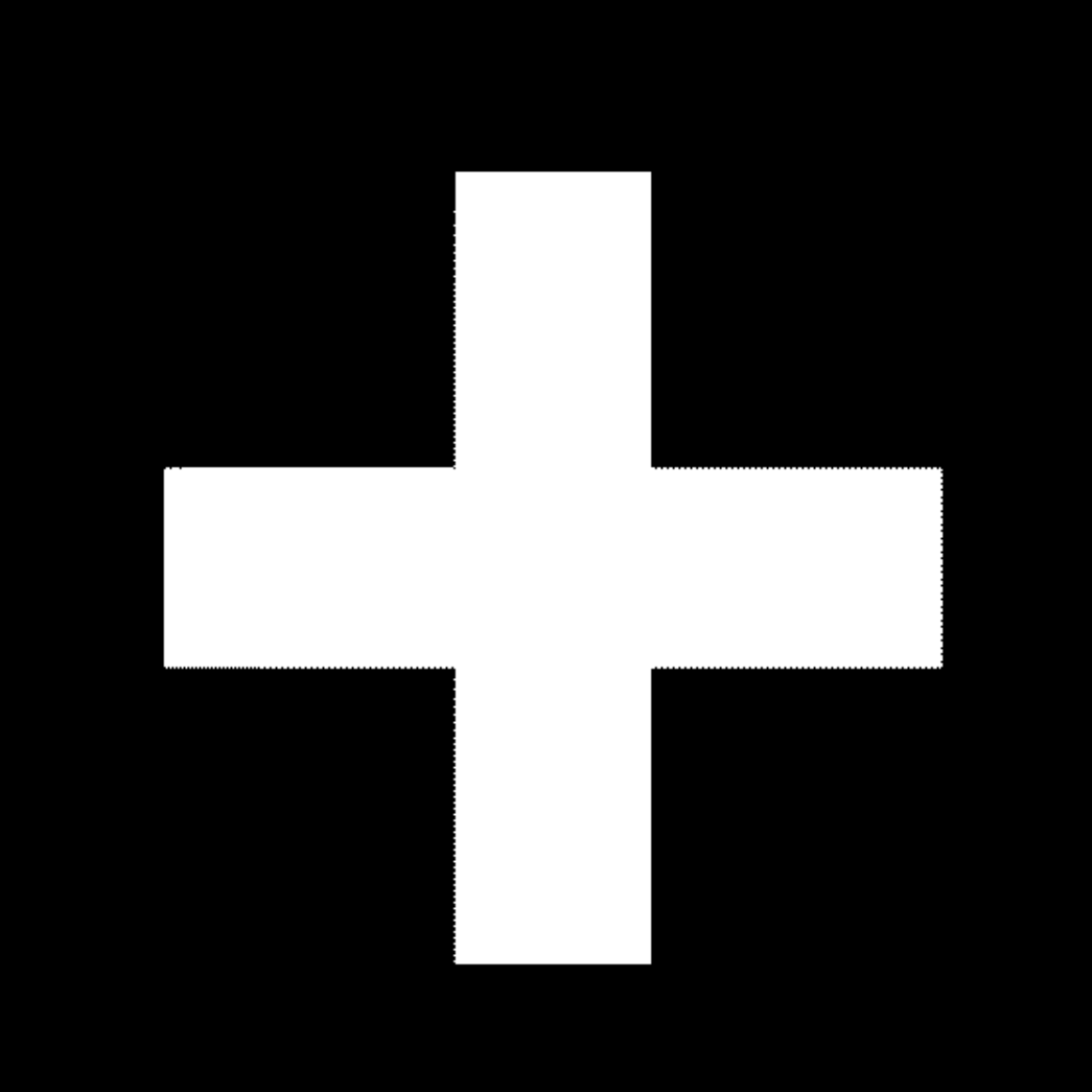}
\end{minipage}
}
\subfigure[device6-20]{
\begin{minipage}[htpb]{0.15\linewidth}
\centering
\includegraphics[width=0.55in]{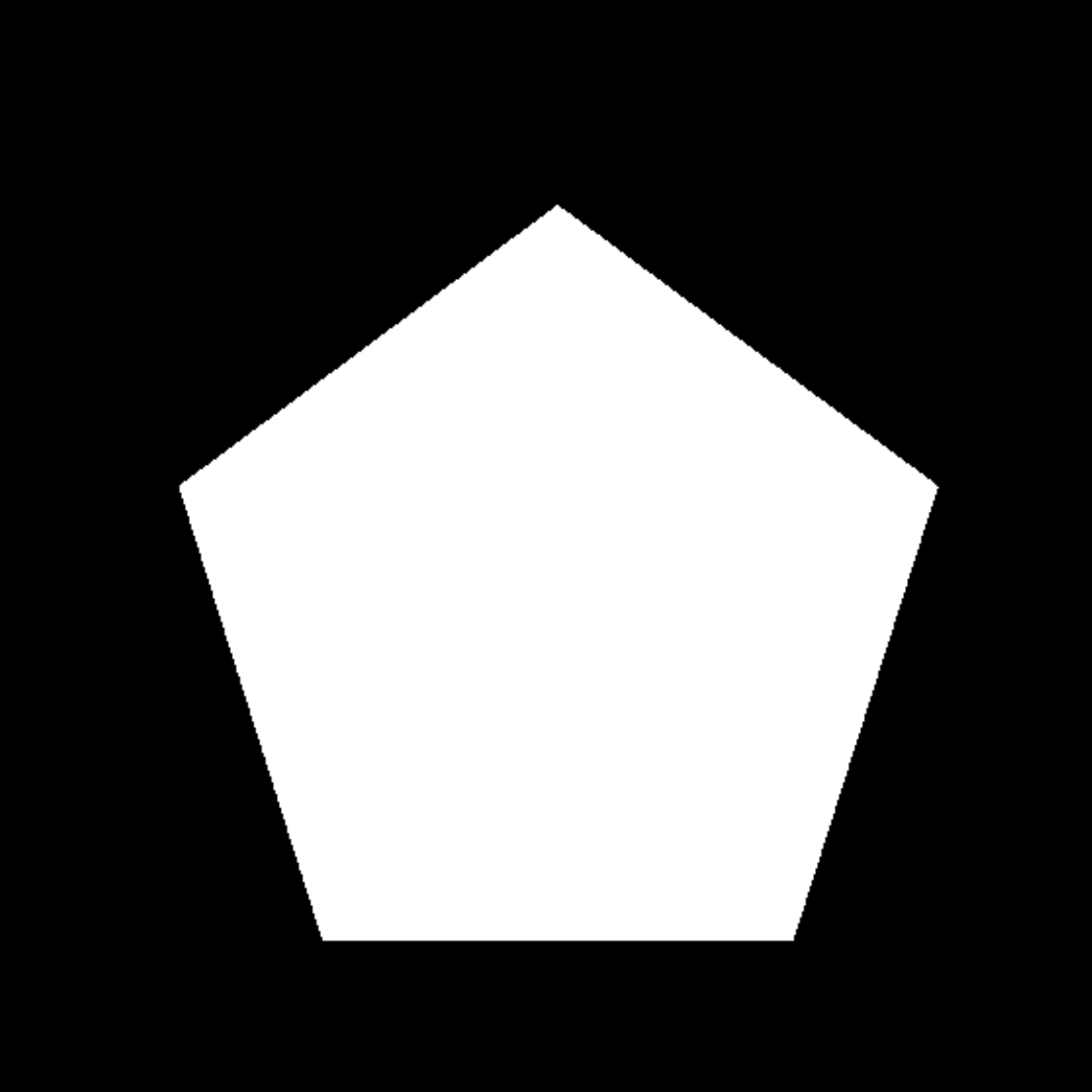}
\end{minipage}
}
\subfigure[device7-19]{
\begin{minipage}[htpb]{0.15\linewidth}
\centering
\includegraphics[width=0.55in]{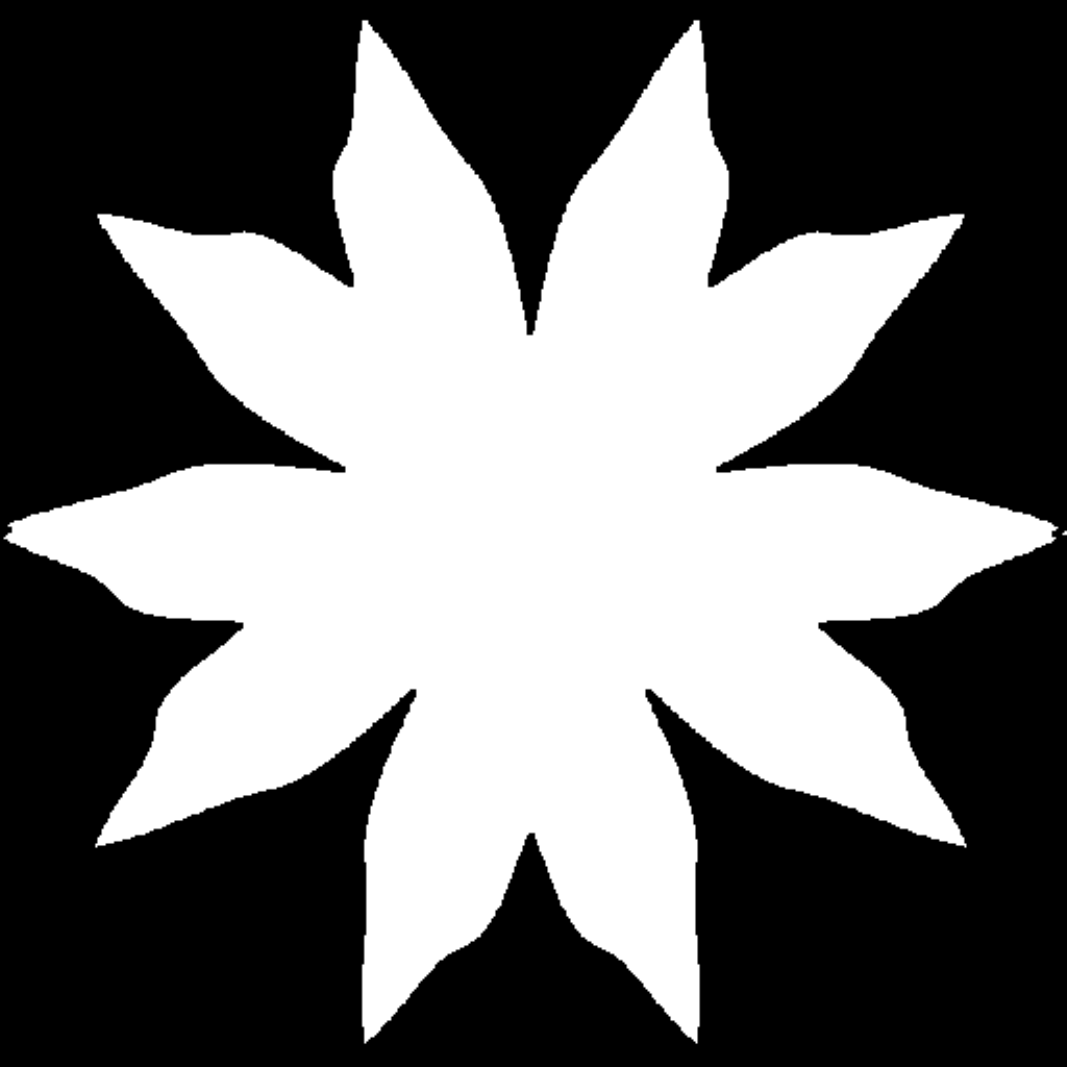}
\end{minipage}
}
\subfigure[device8-20]{
\begin{minipage}[htpb]{0.15\linewidth}
\centering
\includegraphics[width=0.55in]{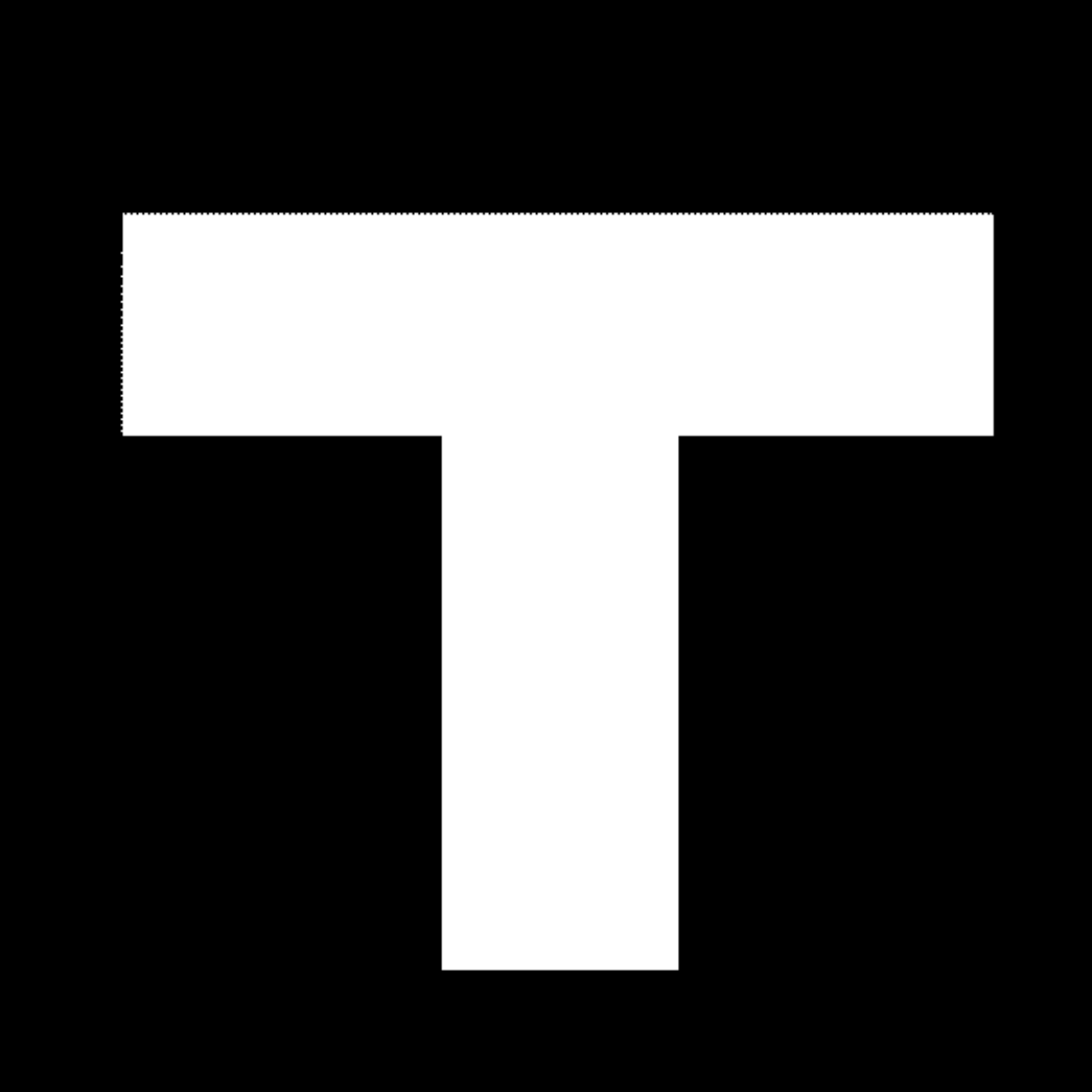}
\end{minipage}
}
\subfigure[device9-20]{
\begin{minipage}[htpb]{0.15\linewidth}
\centering
\includegraphics[width=0.55in]{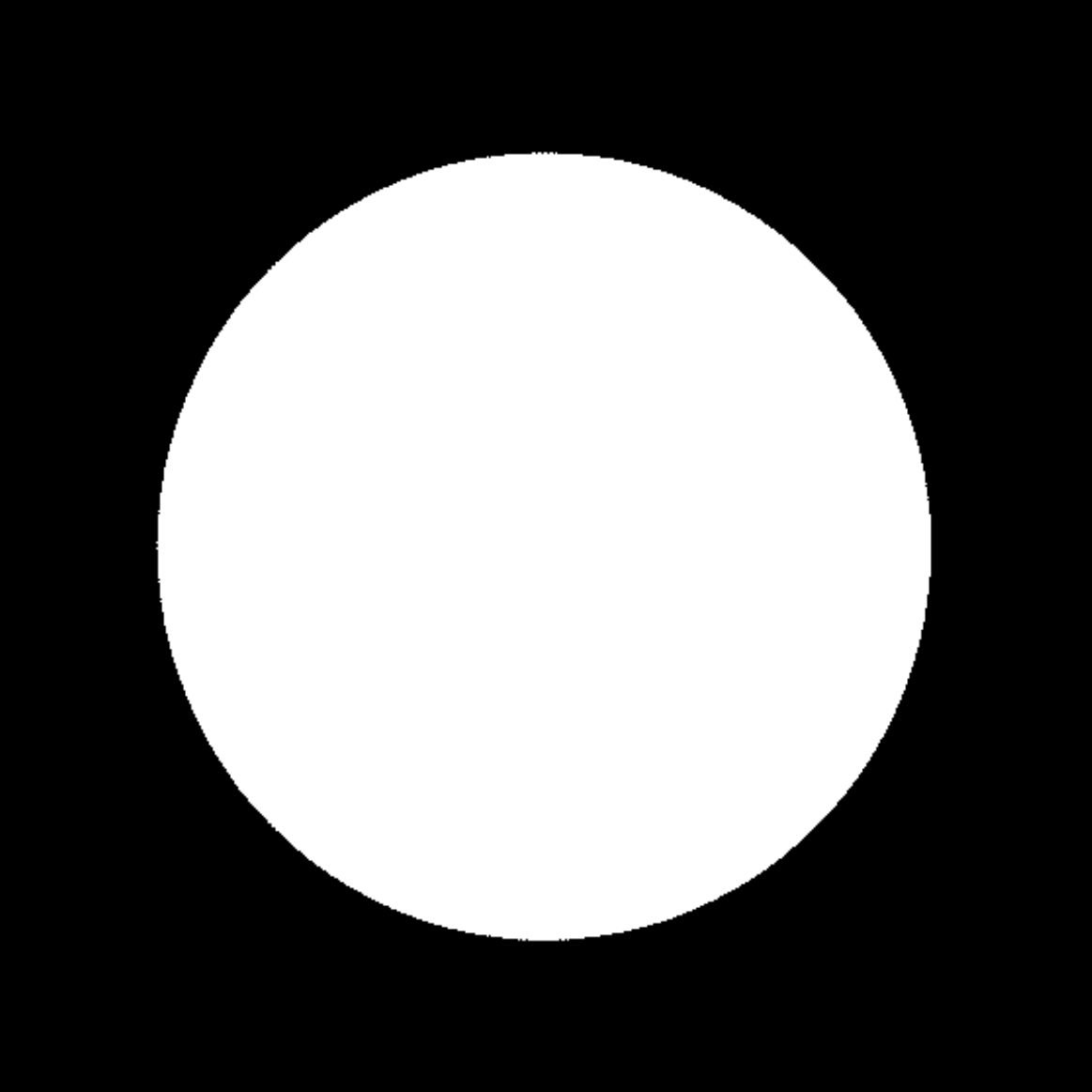}
\end{minipage}
}
\caption{The signal shapes in this paper. These shapes are all black-and-white pictures and are transformed into matrices with elements 0 and 1, where the white parts are represented by 0 and the black parts are 1. These shapes are all $64\times 64$ matrices.}
\end{figure}
\subsection{Signal Shape Data Sets}
This section evaluates the effect of the safe subspace screening rule in Algorithm 2 on some simulation data sets in Figure 1, which was analyzed in \cite{Z14}. Same as this reference, we select some clear signal shapes in Figure as the objective and set them as true solutions $B$, and simulate the prediction matrix $\{X_{i}\}_{i=1}^{n}$ with their elements distributing as the standard normal distribution. The noise term $\{\epsilon_{i}\}_{i=1}^{n}$ distribute as the Gaussian distribution with mean 0 and standard variance 0.1. Following $y_{i}=\langle X_{i},B\rangle+\epsilon_{i}$, response variables can be obtained. In these data sets, we set $K=10$ as in \cite{Z14}.

There results are presented  in Table 2 and Table 3. From these numerical results, we can conclude the following claims. 1) The proposed screening rule surely accelerate the computation of the model. One can see that the speedup values vary with different data and different sample size $n$, meaning the effect of our safe subspace screening rule being different under different data sets. 2) In Table 2 and Table 3, most of $T_{f}$ values are larger or great larger than $T_{s}$, which leads to the speedup values being larger than 1 and can at most being 7.987. The result of device6-20 in Table 3 is special, that is its $T_{f}$ is smaller than $T_{s}$, which leads to its speedup being slightly smaller than 1. So, our screening rule actually speeds up the computation of the  adaptive nuclear norm regularized trace regression in most data sets. 3) With fixed $p$ and $q$, the speedup values decrease with the sample size increasing, which means our screening rule is more efficient on data sets with small size sample number. One can see that the speedup values in Table 2 are more larger than that in Table 3. For instance, the speedup value of device0-9 decreases from $7.987$ to $1.074$ when the sample size increase from 10 to 100.

\begin{table}[htbp]
\caption{Different speedup values of different signal shapes with $p=q=64$, sample size $n=10$ and $n=20$, respectively.}
\begin{tabular}{|l|ccc|ccc|}
\hline
\multicolumn{1}{|c|}{\multirow{2}{*}{Name}} & \multicolumn{3}{c|}{$n=10$}    & \multicolumn{3}{c|}{$n=20$}    \\\cline{2-7}
\multicolumn{1}{|c|}{}                      & $T_{f}$ & $T_{s}$ & speedup & $T_{f}$ & $T_{s}$ & speedup \\\hline
device0-9                                 & 217.017 & 27.173  & 7.987   & 168.720 & 24.435  & 6.905   \\
device1-5                                 & 219.505 & 40.066  & 5.479   & 152.796 & 28.163  & 5.425   \\
device2-1                                 & 195.959 & 39.194  & 4.949   & 125.379 & 28.792  & 4.355   \\
device3-11                                & 194.986 & 35.023  & 5.567   & 169.639 & 29.041  & 5.841   \\
device4-20                                & 245.518 & 34.375  & 7.142   & 135.808 & 19.493  & 6.967   \\
device5-20                                & 163.850 & 30.389  & 5.392   & 118.962 & 22.972  & 5.179   \\
device6-20                                & 215.899 & 53.224  & 4.056   & 178.745 & 25.811  & 6.925   \\
device7-19                                & 218.524 & 59.622  & 3.665   & 125.296 & 35.455  & 3.534   \\
device8-20                                & 254.065 & 38.345  & 6.626   & 125.379 & 28.792  & 4.355   \\
device9-20                                & 192.724 & 37.218  & 5.178   & 118.913 & 24.672  & 4.820
\\ \hline
\end{tabular}
\end{table}
\begin{table}[htbp]
\caption{Different speedup values of different signal shapes with $p=q=64$, sample size $n=50$ and $n=100$, respectively.}
\begin{tabular}{|l|ccc|ccc|}
\hline
\multicolumn{1}{|c|}{\multirow{2}{*}{Name}} & \multicolumn{3}{c|}{$n=50$}    & \multicolumn{3}{c|}{$n=100$}    \\\cline{2-7}
\multicolumn{1}{|c|}{}                      & $T_{f}$ & $T_{s}$ & speedup & $T_{f}$ & $T_{s}$ & speedup \\\hline
device0-9                                 & 140.902 & 42.386 & 3.324   & 68.403 & 63.675 & 1.074   \\
device1-5                                 & 140.997 & 57.596 & 2.448   & 115.800 & 73.381  & 1.578   \\
device2-1                                 & 149.175 & 55.251 & 2.700   & 105.492 & 86.638  & 1.217   \\
device3-11                                & 159.333 & 63.004 & 2.481   & 104.908 & 88.037  & 1.192  \\
device4-20                                & 90.780 & 42.404 & 2.141   & 83.027 & 70.038  & 1.185   \\
device5-20                                & 88.013 & 46.305 & 1.900   & 102.469 & 74.530  & 1.374   \\
device6-20                                & 86.911 & 50.902 & 1.707   & 69.052 & 71.372  & 0.967   \\
device7-19                                & 139.174 & 59.124 & 2.354   & 98.926 & 83.064  & 1.191   \\
device8-20                                & 130.383 & 44.233 & 2.948   & 82.054 & 75.616  & 1.085   \\
device9-20                                & 126.202 & 57.096 & 2.123    & 87.828 & 77.316 & 1.136
\\ \hline
\end{tabular}
\end{table}
\subsection{COVID-19 Data Set}

This section applies the safe screening rule to the COVID-19 data set. The COVID-19 dataset (\cite{M21, W20})  consists of daily measurements related to COVID-19 for 138 countries around the world. This data set records the newly confirmed case in the period June 13, 2020 to July 12, 2020. In addition, this data also includes the 41 COVID-19 related government policies in each day, i.e.,  school-closing, restrictions on gathering, stay-at-home requirement, income support and so on. Each of these policies may have several levels, for example, school closing includes
no closing, recommend closing, require some closing (e.g. just high school) or require all closing, which varied during the 30-day period. Therefore, for every sample, its prediction matrix is  $X\in\mathbb{R}^{41 \times 30}$  and response $y\in\mathbb{R}$. The sample size is $n = 138$.
\begin{table}[htbp]
\caption{COVID-19}
\begin{tabular}{|c|c|c|c|}
\hline
data set & $T_{f}$ & $T_{s}$ & speedup \\
\hline
COVID-19 & 63.177  & 52.311  & 1.208\\
\hline
\end{tabular}
\end{table}

Same as the last section, we record the $T_{f}$, $T_{s}$ and speedup values. To cover more tuning parameters, we set $K=50$ and $\lambda_{m}=0.618^{k}\lambda_{\max}$, $k=1,2,\cdots,50$. From the result in Table 4, we know that the Algorithm 2 reduce the computation time from 63.177 to 52.311, and make a 1.208 speedup value.
\section{Conclusion}
In this paper, we build up the safe subspace screening rule for the adaptive nuclear norm regularized trace regression. The solution of this model is decomposed to the sum of rank one matrixes. With this decomposition and optimal condition of the model, the safe subspace screening rule is proposed. This rule identifies inactive subspace of the solution decomposition and reduce the solution dimension. Numerical results show that our screening rule can reduce the computational time of the model. However, this result fits the data sets with smaller sample size. How to accelerate the computation of this model with larger sample size is a further consideration.
\section*{Acknowledgements}
This work was supported by the National Natural Science Foundation of China (12371322).




\end{document}